\documentclass[10pt, conference, letterpaper]{ieeeconf}
\usepackage{graphicx} % images
\graphicspath{{./images/}} % image path
\usepackage{float} % keep pics exactly where created w. [H]
\usepackage{hyperref} % internal reference links
\usepackage{mathtools} % math
\usepackage{amsmath} % math
\usepackage{amssymb} % symbols
\usepackage{comment}
\usepackage{xcolor}
\usepackage{accents}
\usepackage{booktabs}
\newtheorem{remark}{Remark}
\newtheorem{assumption}{Assumption}
\newtheorem{definition}{Definition}
\newtheorem{theorem}{Theorem}

\let\labelindent\relax
\usepackage{enumitem}

\usepackage{algorithm} % algorithm
\usepackage{algpseudocode} % algorithm

\newcommand{\set}[1]{\mathcal{#1}}
\newcommand{\ubar}[1]{\underaccent{\bar}{#1}}

\IEEEoverridecommandlockouts
\title{An Adaptive-Sampling Control Framework for Constrained Linear Systems\\ with Robust Safety Guarantees}

\author{Spencer Schutz$^{1}$, Charlott Vallon$^{2}$, and Francesco Borrelli$^{1}$
\thanks{$^{1}$S. Schutz and F. Borrelli are with the Mechanical Engineering Department, University of California, Berkeley, Berkeley, CA 94720. $\{$spencer.schutz,fborrelli$\}$@berkeley.edu.}
\thanks{$^{2}$C. Vallon is with the Mechanical Engineering Department, University of California, Santa Barbara, Santa Barbara, CA 93106. cvallon@ucsb.edu.}
\thanks{This work is sponsored by the Department of the Navy, Office of Naval Research ONR N00014-24-2099.}
}

\begin{document}
\maketitle

\begin{abstract}
Adaptive-sampling control balances control performance with resource efficiency. However, existing methods either fail to guarantee robust constraint satisfaction during rate transitions or require computationally expensive online optimization. This paper proposes an adaptive-sampling control framework for linear systems subject to polytopic state and input constraints and bounded additive disturbances. Given a time-varying reference control update rate provided by a reasoner, our framework continuously calculates Model Predictive Control (MPC) update rates that ensure robust constraint satisfaction at all time steps. Offline, robust $M$-step hold control invariance is used to precompute invariant sets for a list of update rates and transition sets between them. Online, these sets are used in real time to guarantee recursive feasibility and finite-time transitions to the reference update rate. The utility of the architecture is demonstrated in a cruise control simulation. 
\end{abstract}

\section{Introduction} \label{sec:intro}
Adaptive-sampling control adjusts the input update rate online to balance performance and the cost of input updates in real-time~\cite{Dorf_Adaptive_Sampling_PID, Henriksson_Adaptive_Sampling_LQR, Xue_Adaptive_Sampling_MPC, Choi_Adaptive_Sampling_MPC, Gomozov_Adaptive_Sampling_MPC}. By updating at high frequency only when necessary, these approaches reduce computation load during less demanding conditions. Decreasing the update frequency can also lessen energy, logistical, and monetary update costs in applications like embedded control platforms, transportation systems, and healthcare. 

Changes to the input update rate must be carefully coordinated. Rigorous safety guarantees (i.e. constraint satisfaction) must be maintained during a rate transition, and the adaptive-sampling algorithm should weigh all costs of a potential change, including those of the switch itself and of the subsequent control performance. Balancing these considerations while minimizing online computational load is challenging.

Existing approaches typically minimize the number of control updates, without ensuring constraint satisfaction. For example, adaptive-sampling PID~\cite{Dorf_Adaptive_Sampling_PID} and LQR~\cite{Henriksson_Adaptive_Sampling_LQR} controllers adjust update rates heuristically based on system behavior, while adaptive-sampling MPC schemes often vary the sampling time without providing safety guarantees or a principled holistic performance tradeoff~\cite{Gomozov_Adaptive_Sampling_MPC,Xue_Adaptive_Sampling_MPC,Choi_Adaptive_Sampling_MPC}. 

Self-triggered MPC selects update rates online but requires solving multiple optimal control problems or a mixed-integer program at each input update~\cite{Zhan_robust_self_trigger,Brunner_robust_self_trigger}, and cannot guarantee convergence to a reference update rate~\cite{Zhan_robust_self_trigger,Brunner_robust_self_trigger,Dai_cost-based_self_trigger}.

Move-blocking MPC constrains the control input to remain constant over predefined intervals within the prediction horizon, but either requires updates at every time step to ensure recursive feasibility~\cite{Behrunani_Multi_Horizon_MPC, Gondhalekar_Safe_Move-block_MPC,Cagienard_Move-block_MPC}, lacks robustness to uncertainty~\cite{Gondhalekar_Least_Restrictive_Move-block_MPC}, or cannot adapt the blocking structure online without re-solving complex optimization problems~\cite{Son_semi-explicit_move-block_MPC}. %These approaches also generally do not provide a mechanism for systematically trading off performance against update rate during operation.

This paper proposes an adaptive-sampling MPC framework for constrained, uncertain linear systems. The algorithm aims to track a reference input update rate while maintaining robust constraint satisfaction. We vary the update rate through multi-step input holds denoted ``$M$-step holds," where $M$ is the number of discrete time steps for which an input must be held constant. Offline, we use robust $M$-step hold control invariance~\cite{Schutz_MSH, Schutz_MPC} to precompute safe states for a list of update rates (values of $M$), along with new robust transition states between them. Online, an MPC controller uses these sets with an update rate chosen by a parameter selection algorithm, solving a single convex quadratic program per input update and avoiding mixed-integer programs. We provide proofs of robust constraint satisfaction, recursive feasibility, and finite-time transition to the reference update rate. We demonstrate the utility of our framework in a cruise control example.

\section{Notation} \label{sec:not}
Let $\mathbb{Z}_{\geq 0}$ and $\mathbb{Z}_{>0}$ denote the non-negative and positive integers, and let $\mathbb{Z}_{a:b}$ denote the integers from $a$ to $b$, inclusive. The map $i\mapsto M\lfloor i/M\rfloor$ rounds $i\in\mathbb{Z}_{\geq 0}$ down to the nearest multiple of $M\in\mathbb{Z}_{>0}$. For a matrix $E$ and sets $\set{X}$ and $\set{W}$, the image, Minkowski sum, and Pontryagin difference~\cite{Kolmanovsky_sets} are $E\set{W} = \{Ew \mid w \in \set{W}\}$, $\set{X}\oplus\set{W} = \{x+w \mid x\in\set{X}, w\in\set{W}\}$, and $\set{X}\ominus\set{W} = \{y \mid y+w\in\set{X}, \forall w\in\set{W}\}$, respectively. For a vector $x$ and symmetric matrix $Q \succeq 0$, $\|x\|_Q^2 = x^\top Q x$.

\section{Problem Formulation}\label{sec:prob}

\subsection{System}
Consider a discrete-time linear time-invariant system model with state $x_t$, control input $u_t$, and uncertainty $w_t$,
\begin{subequations}\label{eqn:dyn_const}
\begin{gather}
    x_{t+1}=Ax_t+Bu_t+Ew_t \label{eqn:dyn}\\
    x_t\in\set{X},~u_t\in\set{U}, ~w_t\in\set{W}_{\delta(t)}, \label{eqn:const}
\end{gather}
\end{subequations}
with 
%time step $t\in\mathbb{Z}_{\geq 0}$ and 
polytopic state constraints $\set{X}$, input constraints $\set{U}$, and uncertainty sets $\set{W}_{\delta(t)}$ (defined below). 

We consider the input subject to aperiodic holds of varying duration, where the duration of the $k$-th input hold is denoted $M_k\in\mathbb{Z}_{>0}$ (see Fig.~\ref{fig:aperiodic}). The time step corresponding to the start of the $k$-th input hold is denoted $\tau_k$, so that $\tau_{k+1}=\tau_k+M_k$.
%and denoted to begin at time $t=\tau_k$ (where $k\in\mathbb{Z}_{\geq 0}$ and we assume $\tau_0 = 0$). 
%input updates are indexed by event $k\in\mathbb{Z}_{\geq 0}$ rather than time step $t$. 
%The time step corresponding to the $k$-th input update is denoted $\tau_k$ (we assume $\tau_0=0$). 
%The length of the $k$-th input hold is given by $M_k\in\mathbb{Z}_{>0}$, such that $\tau_{k+1}=\tau_k+M_k$. 
For all $t$ satisfying $\tau_k\leq t<\tau_{k+1}$ for some $k$, the function $\delta(t)=t-\tau_k$ returns the number of time steps elapsed since the last input update.

%The sequence $(M_k)_{k=0}^\infty$ is considered unknown \textit{a priori}. 

\begin{figure}
    \centering
    \includegraphics[width=1\linewidth]{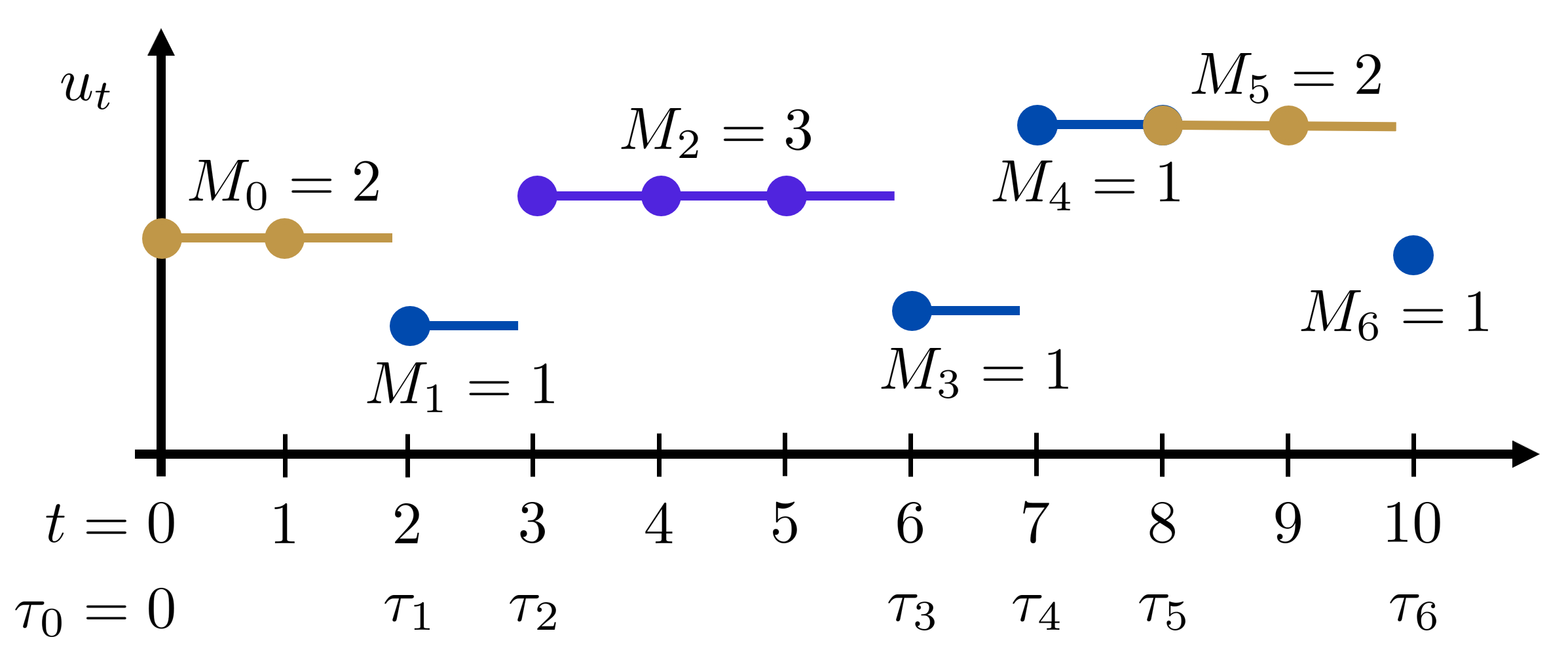}
    \caption{Aperiodic input update time steps $\tau_k$ separated by $M_k$-step holds.}
    \label{fig:aperiodic}
\end{figure}

The uncertainty model $\set{W}_{\delta(t)}$ in~\eqref{eqn:const} is chosen to capture the propagating effects of disturbances and modeling error compared to the real system \cite{STEIN20112626, tomlin2017, hur2019, VOELKER2013943, roy2010}. 
Specifically, we formulate the uncertainty bounds $\set{W}_{\delta(t)}$ to grow as a function of time elapsed since the most recent input update (and simultaneous state measurement). Thus these bounds reset at each time step $\tau_k$.

\begin{assumption}\label{asm:w_nest}
$\set{W}_0$ contains the zero vector and the sets grow monotonically such that $\set{W}_i \subseteq \set{W}_{i+1}$ for all $i \in \mathbb{Z}_{\geq 0}$.
\end{assumption}

\subsection{Robust Adaptive-Sampling MPC}\label{sec:ocp}
% we are living in an adaptive-sampling control world. we want to use an Mk-step hold MPC to control our system (1), which calculates an Mk-step hold input and re-calculates every blah time steps. (maybe something about why MPC is nice and what the high-level formulation features are)

Here, we utilize a robust $M_k$-step hold model predictive control (MPC) policy to control the system~\eqref{eqn:dyn_const}. 
% At every $\tau_k$, the MPC receives a reference hold $M_k^\mathrm{ref}$ from the reasoner, and 
% \begin{enumerate}[label=(\textit{\roman*})]
%     \item solves an optimal control problem (OCP) with parameters selected based on the current state and $M_k^\mathrm{ref}$, and
%     \item applies the first optimal input in an $M_k$-step hold.
% \end{enumerate}
% Specifically, at each $\tau_k$ a hold length $M_k$, OCP horizon $N_k$, and target set $\set{X}_{f,k}$ are selected. 
% %The OCP formulation builds on one introduced in~\cite{Schutz_MPC}. 
% The OCP then searches for $M_k$-step hold inputs that robustly guarantee open-loop constraint satisfaction for the next $M_k$ time steps and inclusion in $\set{X}_{f,k}$ at the next update step. 
% \begin{assumption}\label{asm:N_mod_M}
%     At every $\tau_k$, $N_k$ is a multiple of $M_k$.
% \end{assumption}
$M_k$-step hold MPC adapts the classic MPC formulation by constraining the input sequence of length $N_k$ (the MPC horizon) to satisfy $M_k$-step holds (i.e. only changing every $M_k$ time steps), where $N_k$ is a multiple of $M_k$. 

Building on a formulation introduced in~\cite{Schutz_MPC}, at each $\tau_k$ we solve the following optimal control problem (OCP):
\begin{subequations}\label{eqn:ocp}
\begin{align}
    \min_{\mathbf{\bar{u}}}~ & \|\Delta \bar{x}_{\tau_k+N_k}\|^2_P + \sum_{i=0}^{N_k-1} \big(\|\Delta \bar{x}_{\tau_k+i}\|^2_Q + \|\Delta \bar{u}_{\tau_k+i}\|^2_R\big)  \label{eqn:ocp_cost} \\
    \text{s.t.}~ & \bar{x}_{\tau_k+i+1\mid\tau_k} = A\bar{x}_{\tau_k+i\mid\tau_k} + B\bar{u}_{\tau_k+i\mid\tau_k},~ \forall i \in \mathbb{Z}_{0:N_k-1} \label{eqn:ocp_dyn} \\
    & \bar{x}_{\tau_k+i\mid\tau_k} \in \set{X} \ominus \set{E}_i,~ \forall i \in \mathbb{Z}_{1:M_k-1} \label{eqn:ocp_tight_X} \\
    & \bar{x}_{\tau_k+M_k\mid\tau_k} \in \set{X}_{f,k} \ominus \set{E}_{M_k} \label{eqn:ocp_term_K} \\
    & \bar{x}_{\tau_k+i\mid\tau_k} \in \set{X},~ \forall i \in \mathbb{Z}_{M_k+1:N_k} \label{eqn:ocp_nom_X}\\
    & \bar{x}_{\tau_k\mid\tau_k} = x_{\tau_k} \label{eqn:ocp_init} \\
    & \bar{u}_{\tau_k+i\mid\tau_k} = \bar{u}_{\tau_k+M_k\lfloor i/M_k \rfloor \mid \tau_k} \in \set{U},~ \forall i \in \mathbb{Z}_{0:N_k-1} \label{eqn:ocp_hold}
\end{align}
\end{subequations}
where $\{\bar{x}_{\tau_k+i\mid\tau_k},\bar{u}_{\tau_k+i\mid\tau_k}\}$ denote the predicted nominal states and inputs, and $\mathbf{\bar{u}} = (\bar{u}_{\tau_k+i\mid\tau_k})_{i=0}^{N_k-1}$ is the predicted nominal input sequence. 
Constraints~\eqref{eqn:ocp_tight_X}-\eqref{eqn:ocp_term_K} guarantee robust open-loop state constraint satisfaction for the first $M_k$ steps and inclusion in the target set $\set{X}_{f,k}$ at $\tau_{k}+M_k=\tau_{k+1}$. The target set $\set{X}_{f,k}$ must be chosen to guarantee recursive feasibility and steer the system between regions of the state space where different $M$-step holds are safe.
Robustification is achieved by tightening the constraints using the $i$-step propagated uncertainty set, 
\begin{equation}\label{eqn:Ek}
    \set{E}_i = \bigoplus_{j=0}^{i-1} A^{i-1-j}E \set{W}_{j}.
\end{equation}
Formulating~\eqref{eqn:ocp}-\eqref{eqn:Ek} with nominal dynamics and tightened constraints avoids the computational burden of vertex enumeration~\cite{Borrelli_MPC_book} while preserving robust constraint satisfaction for the uncertain dynamics~\eqref{eqn:dyn_const}~{\cite[Thm.~1]{Schutz_MPC}}.
Since this only relies on the first $M_k$ steps, the tail constraint~\eqref{eqn:ocp_nom_X} serves to bound the predicted nominal cost while avoiding the conservatism of full-horizon robustification. 
Finally, constraint~\eqref{eqn:ocp_hold} enforces input bounds and the $M_k$-step hold over the full $N_k$-step horizon. The objective penalizes $\Delta\bar{x}_{\tau_k+i}=\bar{x}_{\tau_k+i\mid\tau_k}-x^\mathrm{ref}_{\tau_k+i}$ and $\Delta\bar{u}_{\tau_k+i}=\bar{u}_{\tau_k+i\mid\tau_k}-u^\mathrm{ref}_{\tau_k+i}$, deviations from a (freely chosen, possibly unsafe) tracking reference $\{x^\mathrm{ref}_t,u^\mathrm{ref}_t\}$, weighted by $P,Q,R\succeq0$.

We note that at each input update time $\tau_k$, this OCP~\eqref{eqn:ocp}-\eqref{eqn:Ek} is parameterized by a terminal set $\mathcal{X}_{f,k}$, horizon $N_k$, and hold-length $M_k$. 
Once a feasible input sequence is found, the MPC applies the first optimal input for $M_k$ time steps: 
\begin{equation}
    u_{\tau_k} = \pi^{\mathrm{MPC}}(x_{\tau_k}, M_k, N_k, \set{X}_{f,k}) = \bar{u}^\star_{\tau_k\mid\tau_k}.\label{eqn:mpc_pol}
\end{equation}
This input is held constant until the next update at $\tau_{k+1}$. Thus, the MPC applies $u_t=\bar{u}^\star_{\tau_k\mid\tau_k}$ for $\tau_k\leq t < \tau_{k+1}$.

\subsection{Adaptive-Sampling Control Framework}

Adaptive-sampling control considers the online adaptation of the control input update rate $M_k$. We assume that a high-level reasoner provides reference input hold lengths for the MPC controller~\eqref{eqn:mpc_pol}, suggesting how many time steps in the future the OCP~\eqref{eqn:ocp}-\eqref{eqn:Ek} be resolved (Fig.~\ref{fig:reference}). The value of the reference hold at the update time $\tau_k$ is denoted $M_k^{\rm{ref}}$. %Design of the reasoner is beyond the scope of this paper. 
\begin{assumption}\label{asm:m_list}
    The high-level reasoner selects reference hold values $M_k^{\rm{ref}}$ from a predetermined, finite set $\mathbb{M}$. 
\end{assumption}

\begin{figure}
    \centering
    \includegraphics[width=1\linewidth]{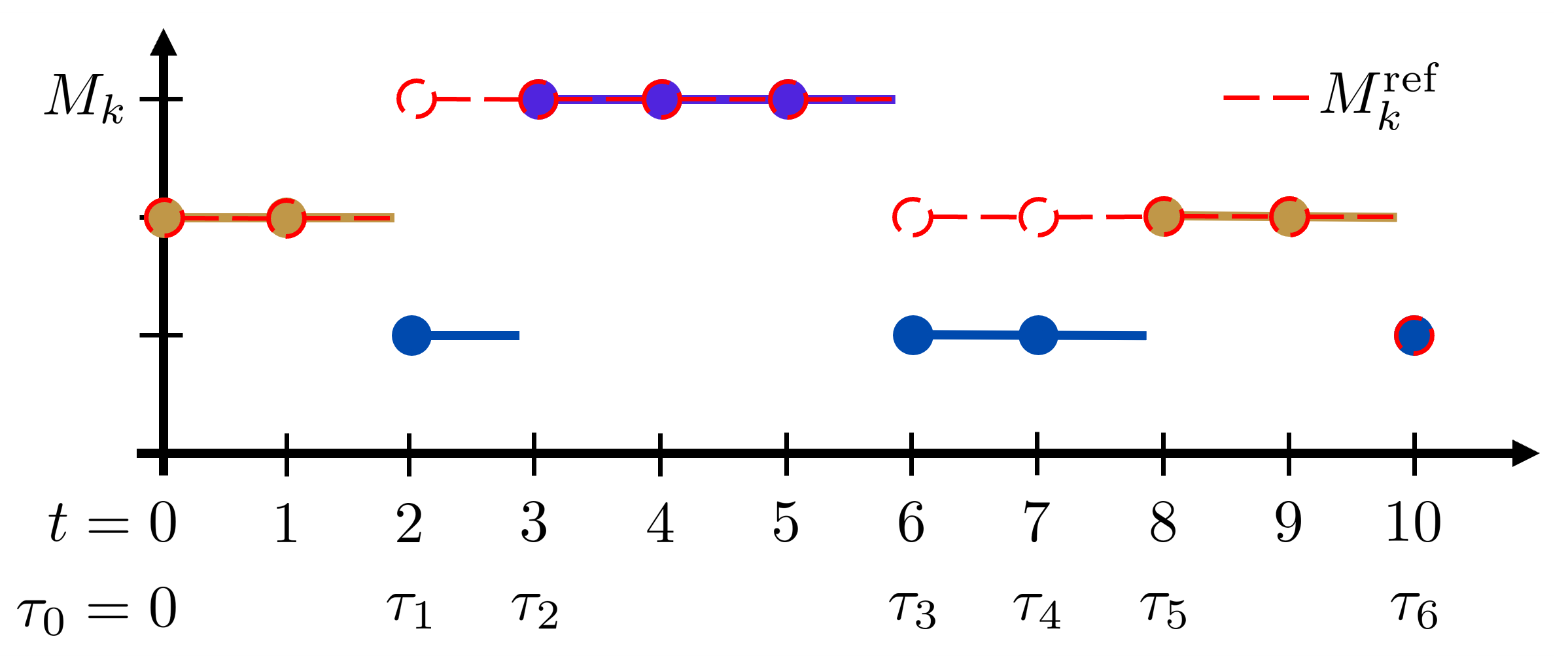}
    \caption{Input hold $M_k$ selected to track the reference hold $M^{\mathrm{ref}}_k$.}
    \label{fig:reference}
\end{figure}

Based on the current system state $x_{\tau_k}$, switching the length of the input hold from $M_{k-1}$ to $M_k^{\rm{ref}}$ may not be immediately possible at time $\tau_k$. For example, if the system is close to a state constraint boundary, it may be unsafe to lower the control update rate until the system has been ``steered" to a new state from which it is safe to apply $M_k^{\rm{ref}}$-step hold MPC. This steering can be implemented via appropriate choices of the OCP~\eqref{eqn:ocp}-\eqref{eqn:Ek} parameters $M_k$, $N_k$, and $\mathcal{X}_{f,k}$ in the MPC policy~\eqref{eqn:mpc_pol} (see Fig.~\ref{fig:architecture}). 
The online selection of these parameters is the focus of this paper. The parameter selection algorithm must be designed so that the resulting robust adaptive-sampling MPC policy:
\begin{enumerate}[label=(\textit{\roman*})]
    \item guarantees robust state and input constraint satisfaction at all time steps $t$;
    \item is recursively feasible at every $\tau_k$; and
    \item converges, under constant $M_{k}^{\rm{ref}}$, to an $M_{k}^{\rm{ref}}$-step hold policy in finite time. 
\end{enumerate}

Offline, we leverage set-theoretic tools from~\cite{Schutz_MSH,Schutz_MPC} to precompute sequences of $\{M,N,\set{X}_f\}$ that allow for safe transitions between pairs of hold values in $\mathbb{M}$.
%We do this by constructing modified invariant sets for different hold lengths and then calculating possible transition sets between them. To address~\textbf{\ref{req:no_future}}, we guarantee transitions may be safely interrupted at an input update step if the reference hold has changed. 
We bound the worst-case number of time steps required to transition and construct a transition graph weighted by a chosen transition cost. These steps are outlined in Sec.~\ref{sec:offline}.

Online, if it is safe at input update $k$ to apply an input hold of length $M_k^{\rm{ref}}$, we set $M_k = M_k^{\rm{ref}}$ (with $N_k$ and $\mathcal{X}_{f,k}$ chosen as detailed in Sec.~\ref{sec:online}). If it is not safe, the parameter selection algorithm uses the transition graph to select a \textit{sequence} of OCP parameter tuples $\{M_k, N_k, \mathcal{X}_{f,k}\}$ that, when implemented, robustly steer the system to a state from which an input hold of length $M_k^{\rm{ref}}$ can be applied. We show that our online selection of $\set{X}_{f,k}$ guarantees recursive feasibility of~\eqref{eqn:ocp}-\eqref{eqn:Ek} at every $\tau_k$, and thus the closed loop system robustly satisfies constraints at all $t$. 
% This sequence determines $\{M_k,N_k,\set{X}_{f,k}\}$, and we solve the corresponding OCP~\eqref{eqn:ocp}-\eqref{eqn:Ek}. 
% This is detailed in Sec.~\ref{sec:online}. 
% If it is safe at update $k$ to apply an input hold of length $M_k^{\rm{ref}}$, we set $M_k = M_k^{\rm{ref}}$ (with $N_k$ and $\mathcal{X}_{f,k}$ chosen as detailed in Sec.~\ref{sec:online}). If it is not safe, the algorithm calculates 

In Sec.~\ref{sec:cc}, we demonstrate the utility of our architecture in a cruise control example.

\begin{figure}
    \centering
    \includegraphics[width=1\linewidth]{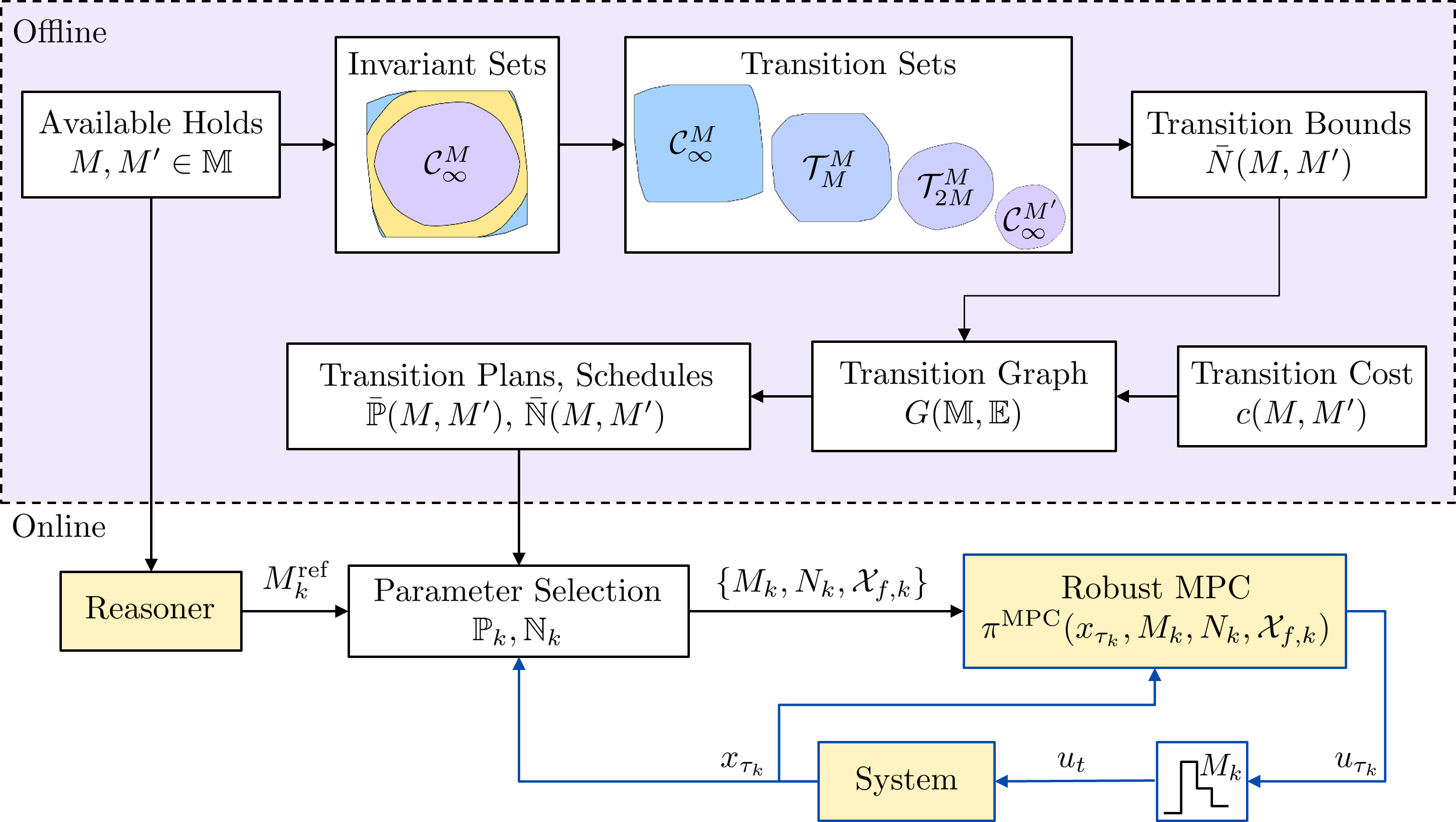}
    \caption{System~\eqref{eqn:dyn_const} in closed-loop with a reasoner and MPC~\eqref{eqn:mpc_pol}. Offline, Alg.~\ref{alg:offline} computes transition plans and schedules between all input holds. Online, Alg.~\ref{alg:online} selects parameters for~\eqref{eqn:ocp}-\eqref{eqn:Ek} to track the reference hold.}
    \label{fig:architecture}
\end{figure}

\section{Offline: Transition Planning}\label{sec:offline}
%This needed more storytelling, so I gave it a shot
Before beginning the control task, we utilize tools from robust $M$-step hold control invariance to calculate sets of states which the system can safely reach under the input holds in $\mathbb{M}$. We connect these into a weighted ``transition graph" that maps how sequences of different input holds can be applied to steer the system between these sets (and into a set from which the desired reference hold is safe to apply). These steps are consolidated in Alg.~\ref{alg:offline}.
%This transition graph will later be used online to prescribes tuples of $\{M,N,\set{X}_{f}\}$ based on $x_{\tau_k}$, $M_k$, and $M_k^{\rm{ref}}$.
% , M^\prime\in\mathbb{M}$ (Alg.~\ref{alg:offline}).

\begin{algorithm}[t] 
\caption{Offline Transition Planning}\label{alg:offline}
\textbf{Input:} System~\eqref{eqn:dyn_const}, available holds $\mathbb{M}$\\
\textbf{Output:} Transition sets, plans, and schedules; pruned $\mathbb{M}$
\begin{algorithmic}[1]
\State Compute $\set{C}^M_\infty$ for all $M\in\mathbb{M}$ (Defs.~\ref{def:msh_ctrl_inv}-\ref{def:msh_pre_set},~\cite[Alg. 2]{Schutz_MSH}).\label{alg:offline_C}
\State Remove any $M$ from $\mathbb{M}$ for which $\set{C}^M_\infty$ is empty.\label{alg:offline_M}
\State Compute transition sets $\set{T}^M_i(\set{C}^{M^\prime}_\infty)$ (Def.~\ref{def:msh_trans_set}) and transition bounds $\bar{N}(M,M^\prime)$~\eqref{eqn:N_bar} between all $M,M^\prime\in\mathbb{M}$.\label{alg:offline_N}
% \State Remove any $M$ from $\mathbb{M}$ that cannot be transitioned into.\label{alg:offline_M2}
\State Tighten transition sets for use in OCP constraint~\eqref{eqn:ocp_term_K}.\label{alg:offline_T}
\State Construct transition graph $G(\mathbb{M},\mathbb{E})$ from transition bounds; verify strong connectivity by Rmk.~\ref{rmk:strong_connect}.\label{alg:offline_G}
\State Assign edge cost $c:\mathbb{M}\times\mathbb{M}\mapsto\mathbb{R}_{\geq 0}$. \label{alg:offline_c}
\State Compute optimal transition plans $\bar{\mathbb{P}}(M,M^\prime)$ and worst-case schedules $\bar{\mathbb{N}}(M,M^\prime)$ between all $M,M^\prime\in\mathbb{M}$.\label{alg:offline_P}
\end{algorithmic}
\end{algorithm}

\subsection{Robust \textit{M}-Step Hold Invariant and Transition Sets}\label{subsec:inv}

We compute a set of states for each available $M$-step hold that guarantees robust constraint satisfaction and recursive feasibility when used as $\set{X}_f$.%, and sequences of sets that also guarantee transitions between hold values when used as $\set{X}_f$.
Robust $M$-step hold control invariance~\cite{Schutz_MSH,Schutz_MPC} extends traditional invariance to systems under a fixed $M$-step hold. %(where $M_k=M$ for all $k\in\mathbb{Z}_{\geq 0}$).

\begin{definition}\label{def:msh_ctrl_inv}
A set $\set{C}^M\subseteq{\set{X}}$ is \textit{\textbf{robust $M$-step hold control invariant}} for system~\eqref{eqn:dyn} with a fixed $M$-step hold and subject to constraints~\eqref{eqn:const} if
\begin{align*}
    x_{t} \in \set{C}^M \Rightarrow \exists~ &u_{t} \in \set{U} ~\text{s.t.}~ \forall w_{t+i} \in \set{W}_{i}, \forall i \in \mathbb{Z}_{0:M-1}, \\
    &x_{t+i+1} = A x_{t+i} + B u_{t} + E w_{t+i} \in \set{X}, \\
    &x_{t+M} \in \set{C}^M.
\end{align*}   
\end{definition}

From any $x_t\in\set{C}^M$, the system can be robustly kept within $\set{X}$ at all time steps, and returned to $\set{C}^M$ every $M$ time steps, under $M$-step hold inputs. An $M$-step hold MPC solving~\eqref{eqn:ocp}-\eqref{eqn:Ek} with constant hold $M$ and target set $\set{C}^M$ is recursively feasible and guarantees robust constraint satisfaction at all time steps. To maximize the region of attraction for the resulting MPC controller, we choose the largest possible target set $\set{C}^M$.

The \textit{\textbf{maximal robust $M$-step hold control invariant set}} $\set{C}^M_\infty$ is the robust $M$-step hold control invariant set containing all $\set{C}^M\subseteq\set{X}$, defining the largest set of states where an $M$-step hold is safe~\cite{Schutz_MSH, Schutz_MPC}. This set can be computed using a fixed-point algorithm with the robust $M$-step hold precursor set.

\begin{definition} \label{def:msh_pre_set}
The \textbf{\textit{robust $M$-step hold precursor set}} $\mathrm{Pre}^M(\set{S})$ to target set $\set{S}\subseteq\set{X}$ for system~\eqref{eqn:dyn} with a fixed $M$-step hold and subject to constraints~\eqref{eqn:const} is
\begin{align*}
    \mathrm{Pre}^M(\set{S}) = \big\{ x_0 \mid ~&\exists~u_0 \in \set{U} ~\text{s.t.}~\forall w_i \in \set{W}_{i}, \forall i \in \mathbb{Z}_{0:M-1}, \\
    &x_{i+1} = A x_i + B u_0 + E w_i \in \set{X}, \\
    &x_M \in \set{S} \big\}.
\end{align*}
\end{definition}
% For all $x_0\in \mathrm{Pre}^M(\set{S})$, there is an $M$-step held input $u_0$ that robustly guarantees $x_1,\dots,x_{M-1}\in\set{X}$ and $x_M\in\set{S}$. 
Calculation details for $\mathrm{Pre}^M(\set{S})$ are given in~\cite{Schutz_MPC}.

The offline Alg.~\ref{alg:offline} first computes $\set{C}^M_\infty$ for all $M\in\mathbb{M}$ (Line~\ref{alg:offline_C}). We then prune $\mathbb{M}$ to the values that yield non-empty $\set{C}^M_\infty$, ensuring a choice of $\set{X}_{f}$ for steady-state operation under each possible hold (Line~\ref{alg:offline_M}).

While an $M$-step hold is safe to use at every $x_t\in\set{C}^M_\infty$, a different hold $M^\prime\in\mathbb{M}$ is not necessarily safe. If $x_t\notin\set{C}^{M^\prime}_\infty$, we must first steer the state into $\set{C}^{M^\prime}_\infty$ under an $M$-step hold using a sequence of appropriately constructed target sets.

\begin{definition} \label{def:msh_trans_set}
The \textbf{\textit{robust $N$-step $M$-step hold transition set}} $\set{T}^M_N(\set{S})$ to a target $\set{S}\subseteq\set{X}$ is defined recursively:
\begin{align*}
    \set{T}^M_0(\set{S})&=\set{S}\cap\set{C}^M_\infty,\\
    \set{T}^M_{i}(\set{S})&=\mathrm{Pre}^M\big(\set{T}^M_{i-M}(\set{S})\big)\cap\set{C}^M_\infty,~
     \forall i\in \{M,2M,\dots,N\}   
\end{align*}
\end{definition}

All $x\in\set{T}^M_N(\set{S})$ may be robustly steered to $\set{S}$ in $N$ time steps through $\frac{N}{M}$ $M$-step holds while remaining in $\set{X}$ at all time steps and returning to $\set{C}^M_\infty$ every $M$ steps.

Transition sets to $\set{C}^{M^\prime}_\infty$
will be used as a sequence of OCP targets $\set{X}_f$ to steer within $\set{C}^M_\infty$ to states where it is safe to apply $M^\prime$. Intersecting with $\set{C}^M_\infty$ in Def.~\ref{def:msh_trans_set} keeps such trajectories in $\set{C}^M_\infty$ at every input update, so the target $\set{X}_f$ may revert to $\set{C}^M_\infty$ if steering to $\set{C}^{M^\prime}_\infty$ is no longer desired. This grants agility when the reference hold changes frequently. % as by~\textbf{\ref{req:no_future}} the controller cannot plan to meet future references. 

\begin{remark}
    The sets in Def.~\ref{def:msh_trans_set} resemble the ``robust $M$-step hold controllable sets" of~\cite[Def. 5]{Schutz_MPC}, which intersect with $\set{X}$ instead of $\set{C}^M_\infty$ and may therefore leave $\set{C}^M_\infty$, rendering a steering maneuver through them unsafe to interrupt. Def.~\ref{def:msh_trans_set} trades set size for this operational flexibility, as intersecting with $\set{C}^M_\infty$ is more restrictive than with $\set{X}$. %Both adhere to~\textbf{\ref{req:no_inter}}, as neither allows an input hold itself to be safely interrupted.
\end{remark} 

\begin{remark}
    Set invariance under periodic re-entry has been defined in other contexts as $k$-recurrence~\cite{Shen_k-recurrence} and $p$-invariance~\cite{Olaru_p-invariance}, but these notions permit intermittent constraint violation.
\end{remark}

\subsection{Transition Graph}\label{subsec:trans}

\subsubsection{Bounding Direct Transitions}
Definition~\ref{def:msh_trans_set} can be used to find the worst-case number of time steps an $M$-step hold controller requires to robustly transition \textit{every} $x\in\set{C}^M_\infty$ to a target $\set{C}^{M^\prime}_\infty$:
\begin{equation}\label{eqn:N_bar}
    \bar{N}(M,M^\prime)=\min\big\{N\in\{0,M,\dots\}\mid\set{C}^M_\infty\subseteq\set{T}^M_N(\set{C}^{M^\prime}_\infty)\big\}.
\end{equation}

We find $\bar{N}(M,M^\prime)$ for all combinations $M,M^\prime\in\mathbb{M}$ and the associated transition sets $\big\{\set{T}^{M}_i(\set{C}^{M^\prime}_\infty)\big\}_{i=0}^{\bar{N}}$ (Line~\ref{alg:offline_N}). If Eqn.~\eqref{eqn:N_bar} does not have a finite solution for a pair $M,M^\prime\in\mathbb{M}$, the transition is infeasible and assigned $\bar{N}(M,M^\prime) = \infty$. 
For all remaining transitions, we store each set tightened as $\set{T}^{M}_i(\set{C}^{M^\prime}_\infty)\ominus\set{E}_M$ to be retrieved online for use in~\eqref{eqn:ocp_term_K} (Line~\ref{alg:offline_T}). 

Note that while the number of hold pairs $(M,M^\prime)$ is combinatoric in the cardinality of $\mathbb{M}$, the number of transition sets to compute is reduced by a nesting property of maximal robust $M$-step hold control invariant sets.

\begin{theorem}[{\cite[Thm.~3]{Schutz_MPC}}]\label{thm:nest}
    For a system~\eqref{eqn:dyn_const} under Asm.~\ref{asm:w_nest} and $M,M^\prime\in\mathbb{Z}_{>0}$, if $M^\prime$ is a factor of $M$, then $\set{C}^M_\infty\subseteq\set{C}^{M^\prime}_\infty$.
    %the maximal robust $M$-step hold control invariant set is a subset of the maximal robust $M^\prime$-step hold control invariant set.
    % \begin{equation*}
    %     \big(\mathrm{mod}(M,M^\prime)=0\big) \Rightarrow \big(\set{C}^M_\infty\subseteq\set{C}^{M^\prime}_\infty\big).
    % \end{equation*}
\end{theorem}
% \begin{proof}
%     Proof in~\cite[Thm. 3]{Schutz_MPC}.
% \end{proof}

It follows from Thm.~\ref{thm:nest} that if $M^\prime$ is a factor of $M$, then the $0$-step transition set $\set{T}^{M}_0(\set{C}^{M^\prime}_\infty)=\set{C}^{M^\prime}_\infty\cap\set{C}^M_\infty=\set{C}^M_\infty$, implying $\bar{N}(M,M^\prime)=0$. Therefore, Line~\ref{alg:offline_N} need only consider combinations for which $M^\prime$ is not a factor of $M$.

\subsubsection{Planning Intermediate Transitions}
We build a directed transition graph $G(\mathbb{M},\mathbb{E})$ with nodes $\mathbb{M}$ and edges $\mathbb{E}$ for the remaining feasible transitions (Line~\ref{alg:offline_G}). % for which Eqn.~\eqref{eqn:N_bar} had a solution $\bar{N}(M,M^\prime) < \infty$. 
% If $\bar{N}(M,M^\prime)=\infty$ for all $M\in\mathbb{M}$, then $M^\prime$ cannot be safely transitioned into from any other hold, and it is removed from $\mathbb{M}$ (Alg.~\ref{alg:offline}, Line~\ref{alg:offline_M2}). Transitions between all remaining $M,M^\prime\in\mathbb{M}$ are possible, if indirectly. 
We use the transition graph to plan and schedule transitions between holds. 

Let a plan $\mathbb{P}(M^{(0)},M^{(L)}) = (M^{(0)}, M^{(1)}, \dots, M^{(L)})$ denote a sequence of $L+1$ holds $M^{(i)}\in\mathbb{M}$
%,~ i\in\mathbb{Z}_{0:L}$ 
that can be applied to transition from $M^{(0)}$ to $M^{(L)}$. The optimal plan $\bar{\mathbb{P}}$ between any two hold lengths is calculated as:
% \begin{align}
%      \bar{\mathbb{P}}(M,M^\prime) = \arg\min_{\mathbb{P},L\in\mathbb{Z}_{>0}}~&\sum_{i=0}^{L-1} c(M^{(i)}, M^{(i+1)})\\
%      \text{s.t.}~&M^{(0)}=M \nonumber\\
%      &M^{(L)}=M^\prime \nonumber\\
%      &(M^{({i})}, M^{(i+1)})\in\mathbb{E},~\forall i\in\mathbb{Z}_{0:L-1}.\nonumber
% \end{align} 
\begin{equation}\label{eqn:plan}
\begin{aligned}
    \bar{\mathbb{P}}(M,M^\prime) ={}&\arg\min_{\mathbb{P},L}~\sum_{i=0}^{L-1} c(M^{(i)},M^{(i+1)})\\
    \text{s.t.}~~&M^{(0)}=M,~~M^{(L)}=M^\prime,\\
    &(M^{(i)},M^{(i+1)})\in\mathbb{E},~\forall i\in\mathbb{Z}_{0:L-1},
\end{aligned}
\end{equation}
where the cost function $c: \mathbb{M} \times \mathbb{M} \to \mathbb{R}_{\geq 0}$ assigns a cost $c(M, M^\prime)$ to each edge, designed to balance relevant priorities such as the numbers of time steps or input updates required to transition to the desired hold (Line~\ref{alg:offline_c}).
Since the transition graph $G(\mathbb{M}, \mathbb{E})$, the transition bounds $\bar{N}(M,M^\prime)$, and the transition costs $c(M,M^\prime)$ are state-independent, we can precompute all optimal plans $\bar{\mathbb{P}}(M,M^\prime)$ offline, before beginning the control task (Line~\ref{alg:offline_P}). 

The stage of plan $\bar{\mathbb{P}}(M,M^\prime)$ applying $M^{(i)}$ requires at most $\bar{N}^{(i+1)}=\bar{N}(M^{(i)},M^{(i+1)})$ steps, so the entire plan is associated with a worst-case schedule $\bar{\mathbb{N}}(M,M^\prime)=(\bar{N}^{(1)},\dots,\bar{N}^{(L)})$ that provides the worst-case number of steps for which we must apply each planned hold to complete the transition.
%Section~\ref{sec:online} utilizes the identities $\bar{\mathbb{P}}(M,M)=(M,M)$ and $\bar{\mathbb{N}}(M,M)=(0)$ to unify the OCP parameter selection rules. 
% To complete the offline phase, we precompute the set tightenings required for OCP constraints~\eqref{eqn:ocp_tight_X}-\eqref{eqn:ocp_term_K}. 

\begin{remark}\label{rmk:factors}
    While not required, selecting $\mathbb{M}$ with common factors leverages Thm.~\ref{thm:nest} to reduce transition costs. For divisibility chains such as $\{1,2,4,8\}$ or $\{1,5,10,20\}$, transitions are only ever needed to increase the hold.
\end{remark}

\begin{remark}\label{rmk:strong_connect}
    To guarantee online transitions to the reference hold in~Sec.~\ref{sec:online}, we require $G(\mathbb{M},\mathbb{E})$ be \textit{strongly connected}, meaning all pairs of $M,M^\prime\in\mathbb{M}$ are mutually reachable, if indirectly. A sufficient condition for strong connectivity is that the greatest common factor $g$ of the entries of $\mathbb{M}$ satisfies $g\in\mathbb{M}$ and $\bar{N}(g,M^\prime)<\infty$ for all $M^\prime\in\mathbb{M}$, since $\bar{N}(M,g)=0$ for all $M$ by Thm.~\ref{thm:nest}. If $G(\mathbb{M},\mathbb{E})$ is not strongly connected, we prune $\mathbb{M}$ to a strongly connected component of $G$ computed by Tarjan's algorithm~\cite{Tarjan_SCC}.
\end{remark}

\section{Online: Parameter Selection Algorithm}\label{sec:online}
Online at each $\tau_k$, given the previously applied hold value $M_{k-1}$ and current reference $M_{k}^{\rm{ref}}$, the parameter selection algorithm retrieves the stored plan $\bar{\mathbb{P}}(M_{k-1}, M_{k}^{\rm{ref}})$ and worst-case schedule $\bar{\mathbb{N}}(M_{k-1}, M_{k}^{\rm{ref}})$ to select $\{M_k,N_k,\set{X}_{f,k}\}$ for the MPC policy~\eqref{eqn:mpc_pol}. We prove that these selection rules (Alg.~\ref{alg:online}) guarantee recursive feasibility, robust constraint satisfaction, and a finite-time transition to a step reference hold. 

\begin{algorithm}[t] 
\caption{Online Parameter Selection and MPC}\label{alg:online}
\textbf{Input:} $x_{\tau_k}$, $M^\mathrm{ref}_k$, Alg.~\ref{alg:offline} outputs\\
\textbf{Output:} $u_{\tau_k}$\\
\textbf{Persist:} hold plan $\mathbb{P}_k$, worst-case schedule $\mathbb{N}_k$
\begin{algorithmic}[1]
\State \textbf{if} $M^\mathrm{ref}_k\neq M^{(L)}$:\label{alg:online_if_start}
\State \hspace{\algorithmicindent} Overwrite $\mathbb{P}_k$, $\mathbb{N}_k$ with offline solutions $\bar{\mathbb{P}}$, $\bar{\mathbb{N}}$~\eqref{eqn:overwrite}
\State \textbf{end if}\label{alg:online_if_stop}
\State Advance $\mathbb{P}_k$ to latest safe hold for $x_{\tau_k}$~\eqref{eqn:ff}\label{alg:online_ff_P}
\State Select hold $M_k$ via~\eqref{eqn:m_k}\label{alg:online_m_k}
\State Advance $\mathbb{N}_k$ to latest safe transition set for $x_{\tau_k}$~\eqref{eqn:eta_k}\label{alg:online_ff_N}
\State Select target set $\set{X}_{f,k}$ via~\eqref{eqn:xf_k}\label{alg:online_xf_k}
\State Choose horizon $N_k$ as any multiple of $M_k$\label{alg:online_n_k}
\State Solve OCP~\eqref{eqn:ocp}-\eqref{eqn:Ek} with $\{M_k,N_k,\set{X}_{f,k}\}$ for input $u_{\tau_k}$\label{alg:online_ocp}
\State $\tau_{k+1} \leftarrow \tau_{k}+M_k$\label{alg:online_repeat}
\end{algorithmic}
\end{algorithm}

\subsection{Online Plan Processing}\label{subsec:process}
% The parameter selection algorithm maintains two variables online: a current plan of $L+1$ holds $\mathbb{P}_k=(M^{(0)},\dots,M^{(L)})$ and the corresponding schedule $\mathbb{N}_k=(N^{(1)},\dots,N^{(L)})$, initialized as $\mathbb{P}_{-1}(M_{-1},M_{-1})=(M_{-1},M_{-1})$ and $\mathbb{N}_{-1}(M_{-1},M_{-1})=(0)$. 

% At every $\tau_k$, if the reference hold has not changed since $\tau_{k-1}$ we reuse the previous plan and schedule, otherwise we retrieve the offline-calculated optimal plan for the desired transition and follow that plan instead (Alg.~\ref{alg:online}, Lines~\ref{alg:online_if_start}-\ref{alg:online_if_stop}):
% \begin{subequations}\label{eqn:overwrite}
% \begin{align}
%     \mathbb{P}_k&\leftarrow \mathbb{P}_{k-1} ~\mathrm{if}~M^\mathrm{ref}_k=M^{(L)}~\mathrm{else}~\bar{\mathbb{P}}(M_{k-1},M^\mathrm{ref}_k),\\
%     \mathbb{N}_k&\leftarrow \mathbb{N}_{k-1} ~\mathrm{if}~M^\mathrm{ref}_k=M^{(L)}~\mathrm{else}~\bar{\mathbb{N}}(M_{k-1},M^\mathrm{ref}_k)
% \end{align}
% \end{subequations}
% This processing ensures $M^{(0)}=M_{k-1}$ and $ M^{(L)}=M^\mathrm{ref}_k$.

The parameter selection algorithm maintains two variables online. The \textit{plan}
$\mathbb{P}_k=(M^{(0)},\dots,M^{(L)})$ lists the $L+1$ holds to be applied in sequence, from the current
hold $M^{(0)}$ to the reference $M^{(L)}=M^\mathrm{ref}_k$, and the \textit{schedule}
$\mathbb{N}_k=(N^{(1)},\dots,N^{(L)})$ bounds the time steps remaining in each stage, where $N^{(i)}$ is the
number of steps left before the hold may advance from $M^{(i-1)}$ to $M^{(i)}$. Both are initialized at
$\tau_0$ as $\mathbb{P}_{-1}=(M_{-1},M_{-1})$ and $\mathbb{N}_{-1}=(0)$.

At every $\tau_k$, if the reference hold has not changed since $\tau_{k-1}$ we reuse the previous plan and
schedule; otherwise the stored plan is irrelevant and we retrieve the offline optimal solution for the
desired transition (Alg.~\ref{alg:online}, Lines~\ref{alg:online_if_start}-\ref{alg:online_if_stop}):
\begin{subequations}\label{eqn:overwrite}
\begin{align}
    \mathbb{P}_k&\leftarrow\begin{cases}
        \mathbb{P}_{k-1}, & M^\mathrm{ref}_k=M^{(L)}\\
        \bar{\mathbb{P}}(M_{k-1},M^\mathrm{ref}_k), & \text{otherwise}
    \end{cases}\\
    \mathbb{N}_k&\leftarrow\begin{cases}
        \mathbb{N}_{k-1}, & M^\mathrm{ref}_k=M^{(L)}\\
        \bar{\mathbb{N}}(M_{k-1},M^\mathrm{ref}_k), & \text{otherwise}
    \end{cases}
\end{align}
\end{subequations}
In either case, $M^{(0)}=M_{k-1}$ and $M^{(L)}=M^\mathrm{ref}_k$: the first holds because $\mathbb{P}_{k-1}$ was
advanced to begin with the applied hold $M_{k-1}$ at the previous update (see~\eqref{eqn:ff}).

The offline plan $\bar{\mathbb{P}}(M_{k-1},M^\mathrm{ref}_k)$ and schedule $\bar{\mathbb{N}}(M_{k-1},M^\mathrm{ref}_k)$ guarantee a safe transition to $M^\mathrm{ref}_k$, but the state known at each input update may shorten it. We use $x_{\tau_k}$ to advance to the latest safe hold in $\mathbb{P}$ (Line~\ref{alg:online_ff_P}). A planned hold is immediately safe if \textit{(i)} it is a factor of the previously applied hold (by Thm.~\ref{thm:nest}) or \textit{(ii)} $x_{\tau_k}$ is by chance already inside the corresponding invariant set. Both conditions are captured by a single operation:
\begin{subequations}\label{eqn:ff}
\begin{align}
    i^\star&=\max\{i\in\mathbb{Z}_{0:L}\mid x_{\tau_k}\in\set{C}^{M^{(i)}}_\infty\},\label{eqn:ff_i}\\
    \mathbb{P}_k&\leftarrow(M^{(i^\star)},\dots,M^{(L)}),\label{eqn:ff_p}\\
    \mathbb{N}_k&\leftarrow(N^{(i^\star+1)},\dots,N^{(L)}).\label{eqn:ff_n}
\end{align}
\end{subequations}
As all polytopic $\set{C}^M_\infty$ are precomputed,~\eqref{eqn:ff_i} only requires inequality tests. The indices of $\mathbb{P}_k$ and $\mathbb{N}_k$ are reset afterwards. If $L=0$, then $M^{(0)}=M^\mathrm{ref}_k$ and $x_{\tau_k}\in\set{C}^{M^\mathrm{ref}_k}_\infty$, so we have reached the reference hold. As future references are unknown, we append a steady-state tail:
\begin{align*}
    \mathbb{P}_k&\leftarrow (M^\mathrm{ref}_k,M^\mathrm{ref}_k),\\
    \mathbb{N}_k&\leftarrow (0).
\end{align*}

\subsection{OCP Parameter Selection}\label{subsec:param}
Given a processed plan $\mathbb{P}_k$ and schedule $\mathbb{N}_k$, we now define selection rules for the OCP parameters $\{M_k,N_k,\set{X}_{f,k}\}$. 

The first entry of $\mathbb{P}_k$ determines the applied hold (Line~\ref{alg:online_m_k}):
\begin{equation}\label{eqn:m_k}
    M_k=M^{(0)}.
\end{equation}

The first entry of $\mathbb{N}_k$ determines $\eta_k$, the number of time steps remaining until applying the next hold $M^{(1)}$ is guaranteed to be safe. The offline bound gives $\eta_k\leq N^{(1)}$, but with $x_{\tau_k}$ now available we advance through the $M_k$-step hold transition sets (Line~\ref{alg:online_ff_N}):
\begin{equation}\label{eqn:eta_k}
    \eta_k= \min\big\{\eta\in\{0,M_k,\dots,N^{(1)}\}\mid x_{\tau_k}\in\set{T}^{M_k}_{\eta}(\set{C}^{M^{(1)}}_\infty)\big\}.
\end{equation}
As all transition sets are precomputed,~\eqref{eqn:eta_k} again requires only inequality tests. We then overwrite $N^{(1)}$ with $\eta^+_k=\max(0,\eta_k-M_k)$, the steps remaining after the current hold, which determines the target set (Line~\ref{alg:online_xf_k}):
\begin{equation}\label{eqn:xf_k}
    \set{X}_{f,k}=\set{T}^{M_k}_{\eta^+_k}(\set{C}^{M^{(1)}}_\infty).
\end{equation}

Equation~\eqref{eqn:xf_k} requires no online computation, as we simply retrieve the corresponding pre-tightened set. The use of $\max(0,\eta_k-M_k)$ unifies the selection logic for transient and steady-state operation. In the steady-state with $M^\mathrm{ref}_k=M_{k-1}$, rules~\eqref{eqn:m_k}-\eqref{eqn:xf_k} return $M_k=M^\mathrm{ref}_k$, $\eta_k=0$, and $\set{X}_{f,k}=\set{T}^{M^\mathrm{ref}_k}_{0}(\set{C}^{M^\mathrm{ref}_k}_\infty)=\set{C}^{M^\mathrm{ref}_k}_\infty$.

Finally, the horizon $N_k$ is selected as any multiple of $M_k$ (Line~\ref{alg:online_n_k}). It need not equal $\eta_k$, as entry into $\set{C}^{M^{(1)}}_\infty$ is guaranteed by selecting shrinking $\set{X}_{f,k}$ and \textit{not} by shrinking $N_k$ itself. Still, $N_k=\max(\eta_k,M_k)$ may be practical when transitioning between holds, as it prevents the OCP from predicting past the switch step under the previous hold.

\subsection{Example: Plan Processing and Parameter Selection}
Consider $M_{k-1}=6$ with $\mathbb{P}_k=(6,4,8)$ and $\mathbb{N}_k=(6,8)$, which uses an intermediate hold because $\bar{N}(6,8)=\infty$. Section~\ref{subsec:feas} shows that $x_{\tau_k}\in\set{C}^6_\infty$ by construction, and also let $x_{\tau_k}\in\set{C}^4_\infty$ by chance. Equation~\eqref{eqn:ff} yields $i^\star=1$, $\mathbb{P}_k=(4,8)$ and $\mathbb{N}_k=(8)$, so $M_k=4$. Also let $x_{\tau_k}\in\set{T}^4_{4}(\set{C}^8_\infty)$, so Eqn.~\eqref{eqn:eta_k} yields $\eta_k=4$ and $\set{X}_{f,k}=\set{T}^4_{0}(\set{C}^8_\infty)=\set{C}^8_\infty$. The state reaches $\set{C}^8_\infty$ at $\tau_k+4$ (one input update), rather than the worst-case $\tau_k+14$ (three input updates).
\begin{remark}
    The state-dependent advancements in Alg.~\ref{alg:online}, Lines~\ref{alg:online_ff_P},~\ref{alg:online_ff_N} are optional. Completion of the offline optimal plan $\bar{\mathbb{P}}(M_{k-1},M^\mathrm{ref}_k)$ is guaranteed on the worst-case schedule $\bar{\mathbb{N}}(M_{k-1},M^\mathrm{ref}_k)$; Eqns.~\eqref{eqn:ff},~\eqref{eqn:eta_k} can only shorten this. 
\end{remark}

\subsection{Proofs: Safe Input Hold Tracking}\label{subsec:feas}
We now show that Algs.~\ref{alg:offline}-\ref{alg:online} achieve the original goals of \textit{(i)} robust constraint satisfaction for all $t$, \textit{(ii)} recursive feasibility at every $\tau_k$, and \textit{(iii)} a finite-time transition to $M^\mathrm{ref}_k$. Only \textit{(iii)} requires a step reference hold; \textit{(i)} and \textit{(ii)} are guaranteed under general reference sequences.
% In the proof of \textit{(iii)}, we assume a step reference hold such that $M^\mathrm{ref}_{k-1}\neq M^\mathrm{ref}_k$ but $M^\mathrm{ref}_{k+i}=M^\mathrm{ref}_{k}$ for all $i\in\mathbb{Z}_{>0}$. However, we show that \textit{(i)} and \textit{(ii)} are still guaranteed under general reference holds. 

\begin{theorem}\label{thm:mpc_guarantees}
    Consider system~\eqref{eqn:dyn_const} with $\mathbb{M}$ processed by Alg.~\ref{alg:offline}, in closed loop with MPC~\eqref{eqn:mpc_pol}, with $\{M_k,N_k,\set{X}_{f,k}\}$ selected at each $\tau_k$ by Alg.~\ref{alg:online}, initialized at $\tau_0$ with $x_0\in\set{C}^{M_{-1}}_\infty$, $\mathbb{P}_{-1}=(M_{-1},M_{-1})$, and $\mathbb{N}_{-1}=(0)$ for some $M_{-1}\in\mathbb{M}$. Under Asms.~\ref{asm:w_nest}-\ref{asm:m_list}, the following properties hold:
    \begin{enumerate}[label=(\textit{\roman*})]
    \item the system robustly satisfies $x_t\in\set{X}$ and $u_t\in\set{U}$ for all $t\geq0$;
    \item the OCP~\eqref{eqn:ocp}-\eqref{eqn:Ek} is recursively feasible at every input update step $\tau_k$; and 
    \item for a step change in reference hold, with $M^\mathrm{ref}_k\neq M^\mathrm{ref}_{k-1}$ and $M^\mathrm{ref}_{k+i}=M^\mathrm{ref}_k$ for all $i\in\mathbb{Z}_{>0}$, there exists some $j\in\mathbb{Z}_{\geq 0}$ such that the applied hold $M_{k+j}=M^\mathrm{ref}_k$.
    \end{enumerate}
\end{theorem}

\begin{proof}
    Assume the OCP~\eqref{eqn:ocp}-\eqref{eqn:Ek} is feasible at $\tau_k$, where Alg.~\ref{alg:online} selects $M_k=M^{(0)}$ and $\set{X}_{f,k}=\set{T}^{M_k}_{\eta^+_k}(\set{C}^{M^{(1)}}_\infty)$ by~\eqref{eqn:m_k}-\eqref{eqn:xf_k}. Constraint~\eqref{eqn:ocp_hold} gives $\bar{u}^\star_{\tau_k\mid\tau_k}\in\set{U}$, and by~\eqref{eqn:ocp_tight_X}-\eqref{eqn:ocp_term_K} and~\cite[Thm.~1]{Schutz_MPC}, holding $\bar{u}^\star_{\tau_k\mid\tau_k}$ for $M_k$ steps from $x_{\tau_k}$ robustly yields $x_t\in\set{X}$ for all $t\in\mathbb{Z}_{\tau_k:\tau_{k+1}-1}$ and $x_{\tau_{k+1}}\in\set{X}_{f,k}\subseteq\set{C}^{M_k}_\infty$, as every set in Def.~\ref{def:msh_trans_set} is intersected with $\set{C}^M_\infty$.

    We show~\eqref{eqn:ocp}-\eqref{eqn:Ek} is then feasible at $\tau_{k+1}$. The 
    plan $\mathbb{P}_{k+1}$ begins with $M^{(0)}=M_k$, whether overwritten 
    by~\eqref{eqn:overwrite} or carried over, so $i=0$ is admissible 
    in~\eqref{eqn:ff_i}, and~\eqref{eqn:ff}-\eqref{eqn:m_k} return an $M_{k+1}$ 
    with $x_{\tau_{k+1}}\in\set{C}^{M_{k+1}}_\infty$. The first entry of $\mathbb{N}_{k+1}$ is either overwritten with the finite offline bound $N^{(1)}=\bar{N}(M_{k+1},M^{(1)})$ or carried over as $N^{(1)}=\eta^+_k$, in which case $i^\star=0$ and $M_{k+1}=M_k$. Either way 
    $x_{\tau_{k+1}}\in\set{T}^{M_{k+1}}_{N^{(1)}}(\set{C}^{M^{(1)}}_\infty)$, in the former case by Rmk.~\ref{rmk:strong_connect} and~\eqref{eqn:N_bar}, and in the latter by $\set{X}_{f,k}$ as chosen by~\eqref{eqn:xf_k}. Thus, $\eta=N^{(1)}$ is admissible in~\eqref{eqn:eta_k}, $\eta_{k+1}$ is well 
    defined, and~\eqref{eqn:xf_k} returns a stored, pre-tightened 
    $\set{X}_{f,k+1}\subseteq\set{C}^{M_{k+1}}_\infty$ with 
    $x_{\tau_{k+1}}\in\mathrm{Pre}^{M_{k+1}}(\set{X}_{f,k+1})$: by 
    Def.~\ref{def:msh_trans_set} if $\eta_{k+1}\geq M_{k+1}$, and by 
    Def.~\ref{def:msh_ctrl_inv} if $\eta_{k+1}=0$, which by maximality of $i^\star$ occurs only in the steady-state tail, where $M^{(1)}=M_{k+1}$ and $\set{X}_{f,k+1}=\set{C}^{M_{k+1}}_\infty$. An $M_{k+1}$-step held input 
    therefore satisfies~\eqref{eqn:ocp_tight_X}-\eqref{eqn:ocp_term_K} 
    and~\eqref{eqn:ocp_hold}, and the nominal tail~\eqref{eqn:ocp_nom_X} follows by 
    continuing within $\set{C}^{M_{k+1}}_\infty$ with $w_t=0$. 

    The same argument holds at the first update $\tau_0$, where 
    $\mathbb{P}_{-1}=(M_{-1},M_{-1})$, $\mathbb{N}_{-1}=(0)$, and 
    $x_0\in\set{C}^{M_{-1}}_\infty$ by assumption. By induction the OCP is 
    feasible at every $\tau_k$ under any reference sequence, proving \textit{(i)} and~\textit{(ii)}.
    
    % For \textit{(iii)}, the plan retrieved at $\tau_k$ is never overwritten under 
    % a step reference, as $M^\mathrm{ref}_{k+i}=M^{(L)}$ for all 
    % $i\in\mathbb{Z}_{>0}$, and so the plan only advances. The $i$-th subset prescribes 
    % $M^{(i-1)}$ to steer to $\set{C}^{M^{(i)}}_\infty$ in at most $\bar{N}^{(i)}$ steps: each update $\tau_{k+1}$ ends in $\set{X}_{f,k}$, 
    % so~\eqref{eqn:eta_k} returns $\eta_{k+1}\leq\eta^+_k$ and the counter falls by 
    % at least $M^{(i-1)}$ per update, until $\eta^+_k=0$ places the state in 
    % $\set{C}^{M^{(i)}}_\infty$ and~\eqref{eqn:ff_i} advances the plan; if $\bar{N}^{(i)}=0$ the plan advances immediately by Thm.~\ref{thm:nest}. Thus, the applied hold reaches $M^\mathrm{ref}_k$ in at most $\sum_{i=1}^L\bar{N}^{(i)}$ steps, and advancements~\eqref{eqn:ff},~\eqref{eqn:eta_k} only accelerate this.

    For \textit{(iii)}, the plan retrieved at $\tau_k$ is never overwritten under 
    a step reference, and so it only advances. The $i$-th stage prescribes 
    $M^{(i-1)}$ to steer to $\set{C}^{M^{(i)}}_\infty$ in at most 
    $\bar{N}^{(i)}$ steps: each update $\tau_{k+j}$ of the stage ends in 
    $\set{X}_{f,k+j}$, so~\eqref{eqn:eta_k} returns 
    $\eta_{k+j+1}\leq\eta^+_{k+j}$ and the counter falls by at least 
    $M^{(i-1)}$ per update, until $\eta^+_{k+j}=0$ places the state in 
    $\set{C}^{M^{(i)}}_\infty$ and~\eqref{eqn:ff_i} advances the plan; if 
    $\bar{N}^{(i)}=0$ the plan advances immediately by Thm.~\ref{thm:nest}. 
    Thus, the applied hold reaches $M^\mathrm{ref}_k$ in at most 
    $\sum_{i=1}^L\bar{N}^{(i)}$ steps, and 
    advancements~\eqref{eqn:ff},~\eqref{eqn:eta_k} only accelerate this.
\end{proof}

\begin{remark}
    The guarantees above rely solely on the OCP constraints. Any objective~\eqref{eqn:ocp_cost} may be used provided~\eqref{eqn:ocp}-\eqref{eqn:Ek} is a convex QP.
\end{remark}

\begin{remark}
    At $\tau_k$, the MPC may use any $N_k$ that is a multiple of $M_k$ because recursive feasibility only relies on the first $M_k$ steps by design of~\eqref{eqn:ocp}-\eqref{eqn:Ek}.
\end{remark}

\section{Example: Cruise Control}\label{sec:cc}

We now present an adaptive-sampling cruise control to demonstrate the method's intuitive reference hold tracking behavior, and applicability to non-standard LTI systems, specifically those with state-dependent uncertainty. 

Consider the ``lead" ($0$) and ``ego" ($1$) vehicles in Fig.~\ref{fig:platoon}, with positions $\{s_{0,t},s_{1,t}\}$ and velocities $\{v_{0,t},v_{1,t}\}$. Each is modeled as a double integrator with sampling time $T_s$, and the input $u_t$ and uncertainty $w_t$ are the ego and lead accelerations, respectively. For state $x_t=[d_t,v_{1,t},v_{0,t}]^\top$ with following distance $d_t=s_{0,t}-s_{1,t}$,
\begin{align}\label{eqn:cc_dyn}
    x_{t+1}
    &= \begin{bmatrix} 1 & -T_s & T_s \\ 0 & 1 & 0 \\ 0 & 0 & 1 \end{bmatrix} x_t+\begin{bmatrix} -\frac{1}{2}T_s^2 \\ T_s \\ 0 \end{bmatrix} u_t + \begin{bmatrix} \frac{1}{2}T_s^2 \\ 0 \\ T_s\end{bmatrix} w_t. %\nonumber \\ &=Ax_t+Bu_t+Ew_t.
\end{align} 
State constraints bound the following distance and ego velocity, and input constraints bound the ego acceleration,
\begin{subequations}\label{eqn:cc_XU}
\begin{align}
    \set{X}&=\{x\mid d_{\mathrm{min}}\leq d\leq d_{\mathrm{max}},~v_{\mathrm{min}}\leq v_1\leq v_{\mathrm{max}}\},  \label{eqn:cc_X}\\
    \set{U}&=\{u\mid u_{\mathrm{min}}\leq u\leq u_{\mathrm{max}}\}. \label{eqn:cc_U}
\end{align}
\end{subequations}

Robust control invariant sets cannot be computed directly for~\eqref{eqn:cc_dyn}-\eqref{eqn:cc_XU}, as the uncontrolled state $v_{0,t}$ cannot be bounded by $u_t$, rendering the precursor iterations empty~\cite{Lefevre_ACC_2016,Kim_ACC_2019}. We overcome this with a state-dependent uncertainty model that imposes the ego's velocity limits on the lead,
\begin{subequations} \label{eqn:cc_W}
\begin{align}
    \set{W}_t(v_{0,t})&=\{w_t \mid \ubar{w}_t\leq w_t \leq \bar{w}_t\},\\
    \ubar{w}_t&= \max \left( \frac{v_{\mathrm{min}} - v_{0,t}}{T_s}, w_{\mathrm{min}} \right),\\
    \bar{w}_t&= \min \left( \frac{v_{\mathrm{max}} - v_{0,t}}{T_s}, w_{\mathrm{max}} \right).
\end{align}
\end{subequations}

For~\eqref{eqn:cc_dyn}-\eqref{eqn:cc_W}, the robust control invariant set is computed offline as a set of slices at discrete values of $v_0\in[v_{\mathrm{min}},v_{\mathrm{max}}]$~\cite{Lefevre_ACC_2016,Kim_ACC_2019}. This technique was extended to robust $M$-step hold control invariant sets in~\cite{Schutz_MPC}, which designed a fixed-$M$-step hold MPC for~\eqref{eqn:cc_dyn}-\eqref{eqn:cc_W} that guaranteed recursive feasibility with a state-dependent target set $\set{C}^M(v_{0,\tau_k})$, obtained at each input update by intersecting the precomputed slices corresponding to $M$-step roll-outs of the worst-case lead trajectories from $v_{0,\tau_k}$~\cite[Alg. 2, Fig. 3]{Schutz_MPC}. 

Here, we extend~\cite{Schutz_MPC} by incorporating robust $M$-step hold transition sets to track a reference input hold. State-dependent invariant sets imply state-dependent transition sets, so Alg.~\ref{alg:offline} cannot be used to precompute the optimal plans and worst-case schedules. Instead, we select $\{M_k,N_k,\set{X}_{f,k}\}$ using $x_{\tau_k}$ and online set computation. We choose a divisibility chain for $\mathbb{M}$, as in Rmk.~\ref{rmk:factors}, to reduce computation because Thm.~\ref{thm:nest} guarantees all hold reductions. 
%We reduce the computational demand by selecting $\mathbb{M}=\{1,2,4,8,16\}$, which automatically guarantees all hold reductions by Thm.~\ref{thm:nest}. 

When increasing the hold, our online search aims to minimize deviation from the reference hold $|M_k-M_k^\mathrm{ref}|$. If $M^\mathrm{ref}_k$ can neither be applied immediately nor transitioned to under $M_{k-1}$, we recursively test the immediate or finite-time transition to the next-smallest hold in $\mathbb{M}$, reverting to $M_{k-1}$ if all fail. Other strategies may prioritize different metrics. For simplicity, we fix $N_k=\max(\mathbb{M})$ for all $k\in\mathbb{Z}_{\geq 0}$. 

We simulate with $T_s=0.1$, $\mathbb{M}=\{1,2,4,8,16\}$, $\{d_\mathrm{min},d_{\mathrm{max}}\}=\{5,100\}$, $\{v_\mathrm{min},v_{\mathrm{max}}\}=\{0,40\}$, $\{u_\mathrm{min},u_{\mathrm{max}}\}=\{w_\mathrm{min},w_{\mathrm{max}}\}=\{-4,4\}$, $Q=P=\mathrm{diag}(10,0,0)$, $R=1$, and $x_0=[70,30,25]^\top$. 
% To emphasize our method's indifference to the cost function, we supplement~\eqref{eqn:ocp_cost} with a jerk term 
% \begin{equation}
%     \sum_{i=0}^{N_k-1}\|\bar{u}_{\tau_k+i\mid\tau_k}-\bar{u}_{\tau_k+i-1\mid\tau_k}\|^2_{R_\Delta}
% \end{equation}
% where $\bar{u}_{\tau_k-1\mid\tau_k}=u_{\tau_k-1}$ is the previous input and $R_\Delta=10$. 
We choose $d^\mathrm{ref}_t=0$ to encourage the ego to minimize the distance.

The simulation result is shown in Fig.~\ref{fig:sim}. The MPC satisfies constraints at all steps, even when the lead performs a full brake ($18$-$24$s). A requested increase in hold is sometimes immediately safe ($14$s), but elsewhere the controller must apply the largest safe hold ($2$-$2.2$, $18$-$20$, $26$-$27.2$s), or revert to the previous one ($22$-$24$s). The full-brake maneuver highlights the intuitive behavior of robust adaptive-sampling MPC: near the invariant set boundaries it may be unsafe to slow the input update rate, and no transition to a safer state may exist when the input and uncertainty have comparable authority (as here, where they are equal). 

\begin{figure}
    \centering
    \includegraphics[width=1\linewidth]{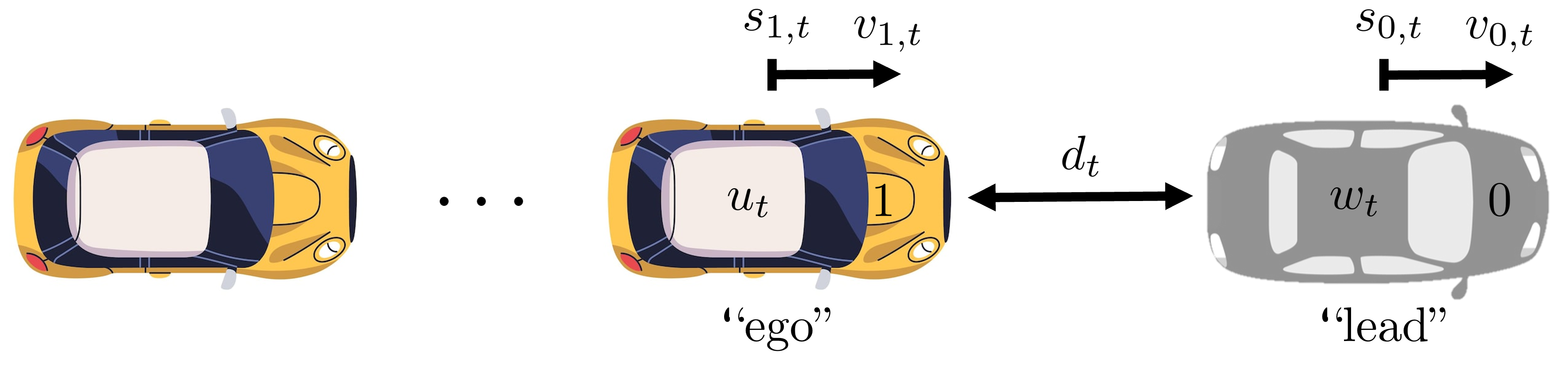}
    \caption{Vehicle platoon: controlled ``ego" ($1$) trails uncontrolled ``lead" ($0$).}
    \label{fig:platoon}
\end{figure}

\begin{figure}
    \centering
    \includegraphics[width=1\linewidth]{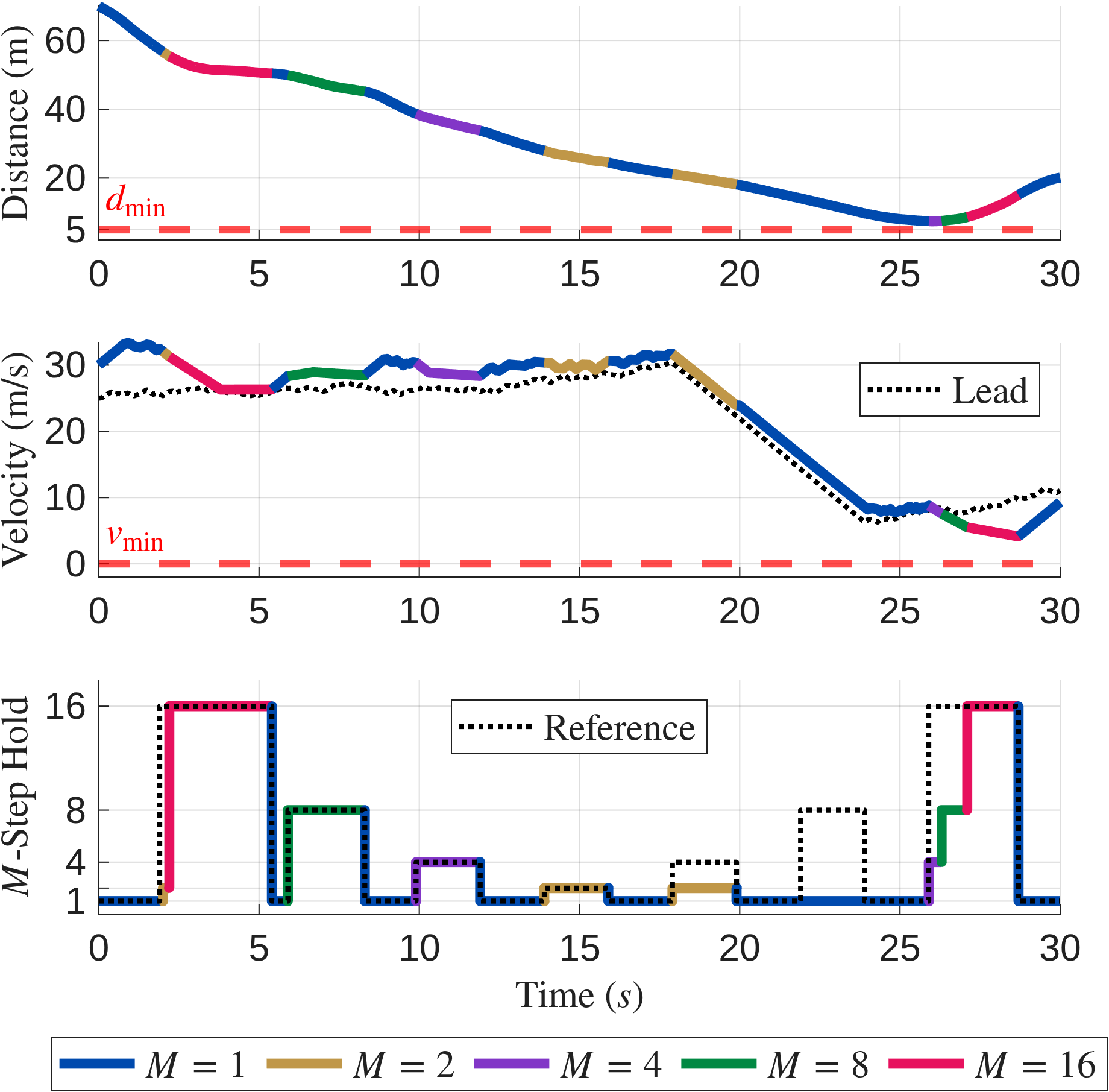}
    \caption{Reference input hold tracking with robust safety guarantees in a cruise control. Online hold selection minimizes deviation from the reference.}
    \label{fig:sim}
\end{figure}

\section{Example II: Lightly Damped Oscillator}\label{sec:osc}

The cruise control of Sec.~\ref{sec:cc} has state-dependent uncertainty bounds, requiring online set computation. We now present a second example with state-independent bounds, for which Algs.~\ref{alg:offline}-\ref{alg:online} apply in full: all sets, transition bounds, plans and schedules are computed offline, and each input update reduces to inequality tests and a single QP.

\subsection{System}

Consider a lightly damped oscillator with damping ratio $\zeta$, natural frequency $\omega_n$, and damped frequency $\omega_d=\omega_n\sqrt{1-\zeta^2}$, sampled with a zero-order hold at $T_s$. The state $x_t=[q_t,\,v_t]^\top$ is position and velocity, and the input $u_t$ and uncertainty $w_t$ are both accelerations. Writing $\rho=e^{-\zeta\omega_nT_s}$ for the contraction and $\theta=\omega_dT_s$ for the rotation over one time step,
\begin{subequations}\label{eqn:osc_dyn}
\begin{gather}
    x_{t+1}=Ax_t+Bu_t+Ew_t,\label{eqn:osc_dyn_a}\\
    A=\rho\begin{bmatrix}
        \cos\theta+\frac{\zeta\omega_n}{\omega_d}\sin\theta & \frac{1}{\omega_d}\sin\theta\\[4pt]
        -\frac{\omega_n^2}{\omega_d}\sin\theta & \cos\theta-\frac{\zeta\omega_n}{\omega_d}\sin\theta
    \end{bmatrix},\label{eqn:osc_A}\\[2pt]
    B=E=\begin{bmatrix}(1-A_{11})/\omega_n^2\\[2pt] A_{12}\end{bmatrix}.\label{eqn:osc_B}
\end{gather}
\end{subequations}
% \begin{subequations}\label{eqn:osc_dyn}
% \begin{gather}
%     x_{t+1}=Ax_t+Bu_t+Ew_t,\label{eqn:osc_dyn_a}\\
%     A=\rho\begin{bmatrix}
%         \cos\theta+\frac{\zeta\omega_n}{\omega_d}\sin\theta & \frac{1}{\omega_d}\sin\theta\\[4pt]
%         -\frac{\omega_n^2}{\omega_d}\sin\theta & \cos\theta-\frac{\zeta\omega_n}{\omega_d}\sin\theta
%     \end{bmatrix},\label{eqn:osc_A}\\[2pt]
%     B=\begin{bmatrix}\big(1-\rho\cos\theta-\rho\frac{\zeta\omega_n}{\omega_d}\sin\theta\big)\big/\omega_n^2\\[4pt]
%         \rho\sin\theta/\omega_d\end{bmatrix},\\
%     E=B.\label{eqn:osc_B}
% \end{gather}
% \end{subequations}
The damped frequency is chosen so that $\theta=\pi/6$, i.e. the oscillator completes one period every $2\pi/\theta=12$ time steps. 
% With the values of Table~\ref{tab:osc_params} this gives
% \begin{equation}\label{eqn:osc_num}
%     A=\begin{bmatrix} 0.8665 & 0.0950\\ -2.6046 & 0.8565\end{bmatrix},\quad
%     B=E=\begin{bmatrix}0.0049\\ 0.0950\end{bmatrix}.
% \end{equation} 
The state and input constraints are
\begin{subequations}\label{eqn:osc_XU}
\begin{align}
    \set{X}&=\{x\mid |q|\leq q_{\mathrm{lim}},~|v|\leq v_{\mathrm{lim}}\},\\
    \set{U}&=\{u\mid |u|\leq u_{\mathrm{lim}}\},
\end{align}
\end{subequations}
and the uncertainty bounds grow as a power of the steps elapsed since the last input update and state measurement,
\begin{equation}\label{eqn:osc_W}
    \set{W}_{\delta}=\{w\mid |w|\leq \bar{w}_0(1+\delta)^{p}\},\quad \delta\in\mathbb{Z}_{\geq0},
\end{equation}
satisfying Asm.~\ref{asm:w_nest}. Parameter values are given in Table~\ref{tab:osc_params}.

\begin{table}[t]
\centering
\caption{Parameters for the lightly damped oscillator.}
\label{tab:osc_params}
\small
\begin{tabular}{@{}c l c @{\hspace{1.2em}} c l c@{}}
\toprule
Sym. & Description & Value & Sym. & Description & Value\\
\midrule
$T_s$    & sampling time      & $0.1$   & $v_{\mathrm{lim}}$ & vel. limit             & $2$\\
$\zeta$   & damping ratio      & $0.01$  & $u_{\mathrm{lim}}$        & acc. limit            & $1.3$\\
$\theta$ & rotation/step  & $\pi/6$ & $\bar{w}_0$       & base uncert. & $0.03$\\
$q_{\mathrm{lim}}$ & pos. limit    & $0.45$  & $p$               & uncert. exp.       & $1.4$\\
\midrule
\multicolumn{6}{@{}l@{}}{$\mathbb{M}=\{1,2,3,4,6,8,10,12\}$, candidate holds}\\
\bottomrule
\end{tabular}
\end{table}

\subsection{Offline: Sets and Transition Graph}

Table~\ref{tab:osc_vol} reports $\mathrm{vol}(\set{C}^M_\infty)$ for each $M\in\mathbb{M}$. Feasibility is \textit{not} monotone in the hold length: $\set{C}^6_\infty$ and $\set{C}^{12}_\infty$ are empty while $\set{C}^8_\infty$ and $\set{C}^{10}_\infty$ are not, so it is safe to hold the input for ten time steps but not for six. For system~\eqref{eqn:osc_dyn}, a held input can only shift the center of the oscillation, not its amplitude, so its effect over a hold depends on where that hold ends in the $12$-step cycle rather than on its length. Alg.~\ref{alg:offline}, Line~\ref{alg:offline_M} prunes $\mathbb{M}$ to $\{1,2,3,4,8,10\}$; these $\set{C}^M_\infty$ are shown in Fig.~\ref{fig:osc_sets}.

\begin{table}[t]
\centering
\caption{Volume of $\set{C}^M_\infty$ for each available hold, with $\mathrm{vol}(\set{X})=3.60$.}
\label{tab:osc_vol}
\begin{tabular}{r rrrrrrrr}
\toprule
$M$            & $1$ & $2$ & $3$ & $4$ & $6$ & $8$ & $10$ & $12$ \\
\midrule
$\mathrm{vol}$ & $2.98$ & $2.94$ & $2.91$ & $2.86$ & --- & $2.78$ & $2.49$ & --- \\
\bottomrule
\end{tabular}
\end{table}

\begin{figure}[t]
    \centering
    \includegraphics[width=1\linewidth]{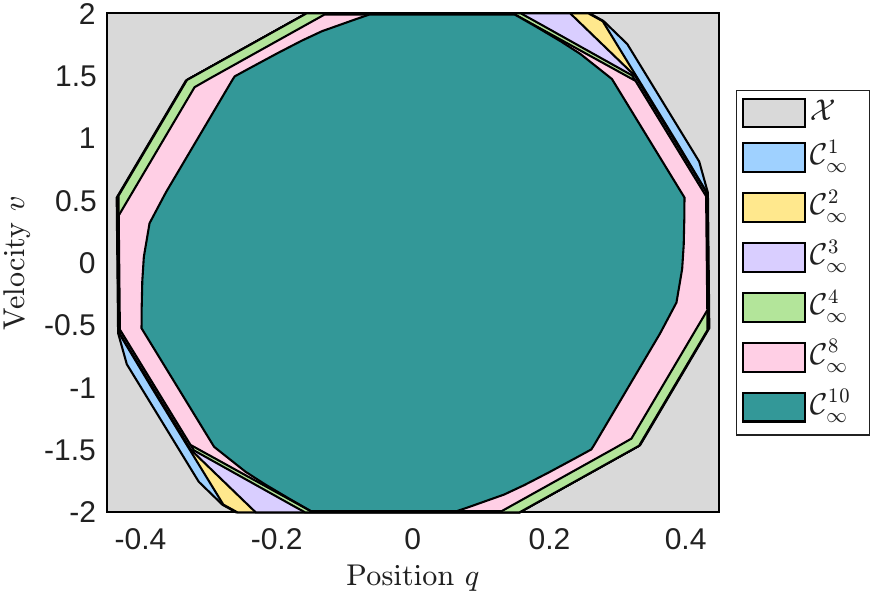}
    \caption{Maximal robust $M$-step hold control invariant sets for system~\eqref{eqn:osc_dyn}.}
    \label{fig:osc_sets}
\end{figure}

Table~\ref{tab:osc_N} lists the transition bounds $\bar{N}(M,M^\prime)$. Every reduction in hold length is immediate, either by Thm.~\ref{thm:nest} or direct verification, so steering is only required to \textit{increase} the hold. The graph $G(\mathbb{M},\mathbb{E})$ is strongly connected, since $\gcd(\mathbb{M})=1\in\mathbb{M}$ and $\bar{N}(1,M^\prime)<\infty$ for all $M^\prime\in\mathbb{M}$.

\begin{table}[t]
\centering
\caption{Transition bounds $\bar{N}(M,M^\prime)$, from $M$ (rows) to $M^\prime$ (columns).}
\label{tab:osc_N}
\begin{tabular}{r rrrrrr}
\toprule
    & $1$ & $2$ & $3$ & $4$ & $8$ & $10$ \\
\midrule
$1$ & $0$ & $4$ & $4$ & $4$ & $6$ & $6$ \\
$2$ & $0$ & $0$ & $4$ & $4$ & $6$ & $6$ \\
$3$ & $0$ & $0$ & $0$ & $3$ & $6$ & $6$ \\
$4$ & $0$ & $0$ & $0$ & $0$ & $8$ & $8$ \\
$8$ & $0$ & $0$ & $0$ & $0$ & $0$ & $16$ \\
$10$ & $0$ & $0$ & $0$ & $0$ & $0$ & $0$ \\
\bottomrule
\end{tabular}
\end{table}

We compare two different transition costs,
\begin{subequations}\label{eqn:osc_costs}
\begin{align}
    c_{\mathrm{steps}}(M,M^\prime)&=\bar{N}(M,M^\prime)+\epsilon\\
    c_{\mathrm{updates}}(M,M^\prime)&=\bar{N}(M,M^\prime)/M+\epsilon,
\end{align}
\end{subequations}
penalizing respectively the time steps and the input updates spent in transition, with a small fixed $\epsilon\ll1$ added to each edge to break ties in favor of fewer hold changes. The two yield different optimal plans (Table~\ref{tab:osc_plans}). Under $c_{\mathrm{steps}}$, transitions into a long hold are routed through $M=1$: reducing the hold is free by Thm.~\ref{thm:nest}, and the finest hold reaches $\set{C}^{10}_\infty$ in $6$ steps rather than the $16$ steps required directly from $M=8$. Under $c_{\mathrm{updates}}$ the direct edge wins instead, since those $16$ steps span only two input updates. Lengthening the hold quickly therefore requires temporarily shortening it, an option the transition graph exposes automatically.

\begin{table}[t]
\centering
\caption{Optimal plans for two transitions under costs~\eqref{eqn:osc_costs}.}
\label{tab:osc_plans}
\small
\begin{tabular}{@{}c ccc ccc@{}}
\toprule
& \multicolumn{3}{c}{$c_{\mathrm{steps}}$} & \multicolumn{3}{c}{$c_{\mathrm{updates}}$}\\
\cmidrule(lr){2-4}\cmidrule(lr){5-7}
$(M,M^\prime)$ & $\bar{\mathbb{P}}$ & stp. & upd. & $\bar{\mathbb{P}}$ & stp. & upd.\\
\midrule
$(4,10)$  & $(4,1,10)$ & $6$ & $6$ & $(4,10)$  & $8$  & $2$\\
$(8,10)$  & $(8,1,10)$ & $6$ & $6$ & $(8,10)$  & $16$ & $2$\\
\bottomrule
\end{tabular}
\end{table}

\subsection{Online: Reference Hold Tracking}

We simulate the closed loop under $c_{\mathrm{steps}}$, with the objective~\eqref{eqn:ocp_cost} tracking a sustained oscillation of~\eqref{eqn:osc_dyn} whose amplitude exceeds the velocity bound of $\set{X}$ ($\max(|v^\mathrm{ref}|) = 2.3$). Such a reference is admissible because the guarantees of Thm.~\ref{thm:mpc_guarantees} rest on the OCP constraints alone; the controller tracks it only as closely as the active $\set{C}^{M_k}_\infty$ permits. The disturbance is realized at the vertex of $\set{W}_{\delta(t)}$ driving the state furthest from the invariant set of the next planned hold.

The result is shown in Fig.~\ref{fig:osc_sim}. Constraints are satisfied at \textit{all} time steps, not only at input updates, using $19$ input updates over $80$ time steps. Applied holds are never interrupted, and the advancements~\eqref{eqn:ff} and~\eqref{eqn:eta_k} shorten transitions relative to their worst-case schedules: the transition from $M=4$ to $M=10$ completes in one time step, against the six steps scheduled for the retrieved plan $\bar{\mathbb{P}}(4,10)=(4,1,10)$. The brief reduction to $M=1$ near $t=2$\,s is that indirect plan being executed: the hold is momentarily shortened in order to lengthen it sooner.

\begin{figure}[t]
    \centering
    \includegraphics[width=1\linewidth]{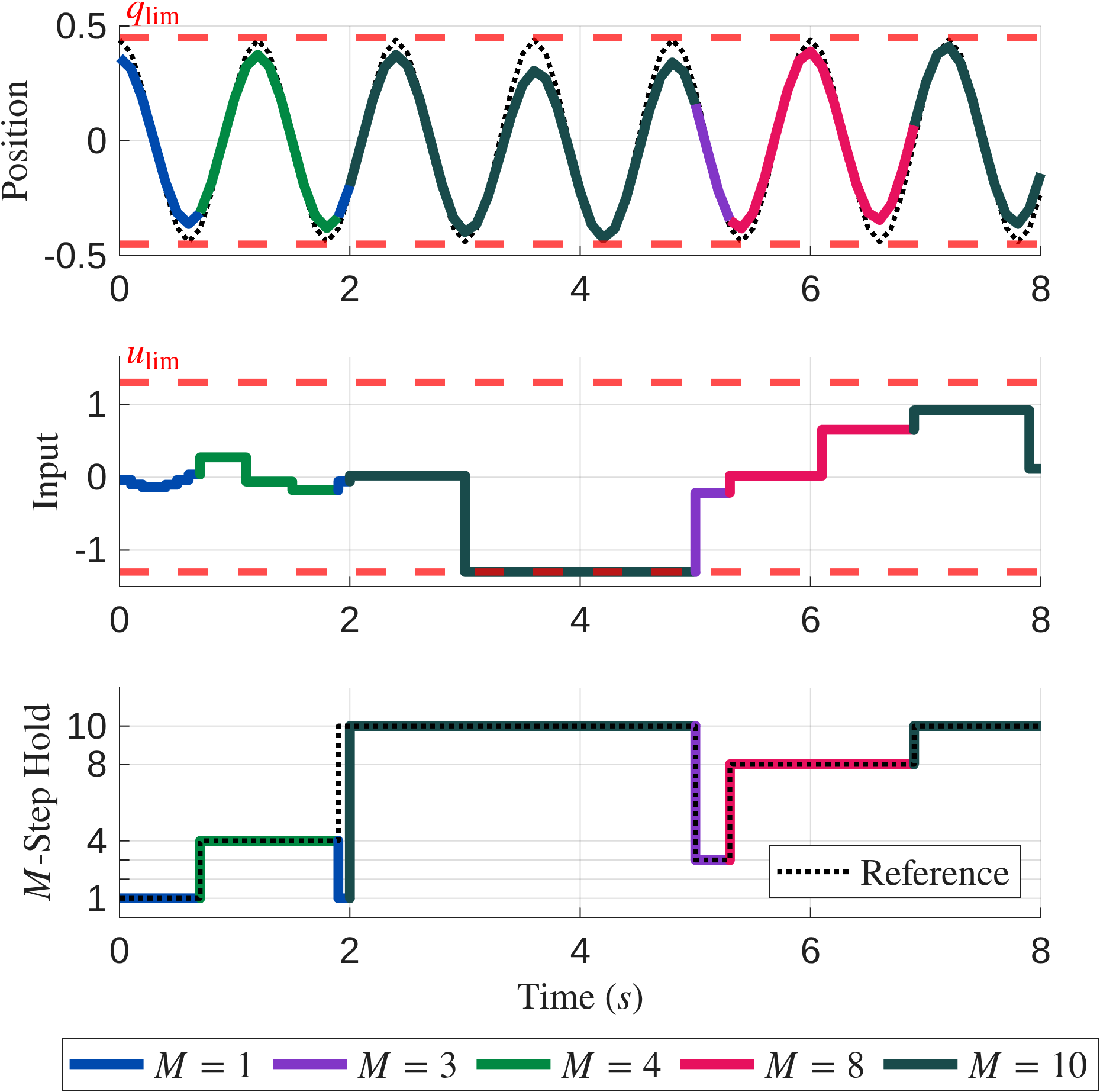}
    \caption{Closed loop for~\eqref{eqn:osc_dyn}-\eqref{eqn:osc_W} tracking a reference hold $M^\mathrm{ref}_k$.}
    \label{fig:osc_sim}
\end{figure}

\section{Conclusion}\label{sec:conc}

This paper presented a framework for adaptive-sampling control of linear systems. We guaranteed robust constraint satisfaction during changes in the input update rate by integrating robust $M$-step hold invariance with a novel transition planning architecture. Offline set computations ensure the online MPC must only solve a single convex quadratic program per update, avoiding a mixed-integer program. We proved robust constraint satisfaction, recursive feasibility, and finite-time transition to the reference hold. A cruise control case study demonstrated the framework's utility.

\bibliographystyle{ieeetr}
\bibliography{root}

@book{Borrelli_MPC_book,
    author = {Borrelli, Francesco and Bemporad, Alberto and Morari, Manfred},
    title = {Predictive Control for Linear and Hybrid Systems},
    publisher = {Cambridge University Press},
    year = {2017}
}

@article{Kolmanovsky_sets,
author = {Kolmanovsky, Ilya and Gilbert, Elmer G.},
title = {Theory and computation of disturbance invariant sets for discrete-time linear systems},
journal = {Mathematical Problems in Engineering},
volume = {4},
number = {4},
pages = {934097},
year = {1998}
}

@ARTICLE{Schutz_MSH,
  author={Schutz, Spencer and Vallon, Charlott and Recht, Benjamin and Borrelli, Francesco},
  journal={IEEE Control Systems Letters}, 
  title={On Sampling Time and Invariance}, 
  year={2025},
  volume={9},
  number={},
  pages={919-924},
}

@ARTICLE{Schutz_MPC,
  title={Safe Adaptive-Sampling Control via Robust M-Step Hold Model Predictive Control}, 
  author={Spencer Schutz and Charlott Vallon and Francesco Borrelli},
  year={2026},
  note={arXiv preprint arXiv:2603.27374.}
}

@article{Xue_Adaptive_Sampling_MPC,
  title={Active collision avoidance system design based on model predictive control with varying sampling time},
  author={Xue, Wanying and Zheng, Ling},
  journal={Automotive innovation},
  volume={3},
  number={1},
  pages={62--72},
  year={2020},
  publisher={Springer}
}

@article{Gomozov_Adaptive_Sampling_MPC,
  title={Adaptive energy management system based on a real-time model predictive control with nonuniform sampling time for multiple energy storage electric vehicle},
  author={Gomozov, Oleg and Trov{\~a}o, Jo{\~a}o Pedro F and Kestelyn, Xavier and Dubois, Maxime R},
  journal={IEEE Transactions on Vehicular Technology},
  volume={66},
  number={7},
  pages={5520--5530},
  year={2016},
  publisher={IEEE}
}

@Article{Choi_Adaptive_Sampling_MPC,
AUTHOR = {Choi, Yoonsuk and Lee, Wonwoo and Kim, Jeesu and Yoo, Jinwoo},
TITLE = {A Variable-Sampling Time Model Predictive Control Algorithm for Improving Path-Tracking Performance of a Vehicle},
JOURNAL = {Sensors},
VOLUME = {21},
YEAR = {2021},
NUMBER = {20},
ARTICLE-NUMBER = {6845},
URL = {https://www.mdpi.com/1424-8220/21/20/6845},
PubMedID = {34696056},
ISSN = {1424-8220},
DOI = {10.3390/s21206845}
}

@article{Dorf_Adaptive_Sampling_PID,
  title={Adaptive sampling frequency for sampled-data control systems},
  author={Dorf, RC and Farren, M and Phillips, C},
  journal={IRE Transactions on Automatic Control},
  volume={7},
  number={1},
  pages={38--47},
  year={1962},
  publisher={IEEE}
}

@inproceedings{Henriksson_Adaptive_Sampling_LQR,
  title={Optimal on-line sampling period assignment for real-time control tasks based on plant state information},
  author={Henriksson, Dan and Cervin, Anton},
  booktitle={Proceedings of the 44th IEEE CDC},
  pages={4469--4474},
  year={2005},
  organization={IEEE}
}

@INPROCEEDINGS{Behrunani_Multi_Horizon_MPC,
  author={Behrunani, Varsha N. and Heer, Philipp and Smith, Roy S. and Lygeros, John},
  booktitle={2024 IEEE 63rd Conference on Decision and Control (CDC)}, 
  title={Recursive feasibility guarantees in multi-horizon MPC}, 
  year={2024},
  volume={},
  number={},
  pages={297-302},
  doi={10.1109/CDC56724.2024.10886145}}

@inproceedings{Shen_k-recurrence,
  author={Shen, Yue and Bichuch, Maxim and Mallada, Enrique},
  booktitle={2022 IEEE 61st Conference on Decision and Control (CDC)}, 
  title={Model-free Learning of Regions of Attraction via Recurrent Sets}, 
  year={2022},
  volume={},
  number={},
  pages={4714-4719},
}

@article{Olaru_p-invariance,
  author={Olaru, Sorin and Soyer, Martin and Zhao, Zhixin and Dórea, Carlos E. T. and Kofman, Ernesto and Girard, Antoine},
  journal={IEEE Transactions on Automatic Control}, 
  title={From Relaxed Constraint Satisfaction to p-Invariance of Sets}, 
  year={2024},
  volume={69},
  number={10},
  pages={7036-7042},
}

@article{Cagienard_Move-block_MPC,
title = {Move blocking strategies in receding horizon control},
journal = {Journal of Process Control},
volume = {17},
number = {6},
pages = {563-570},
year = {2007},
issn = {0959-1524},
author = {R. Cagienard and P. Grieder and E.C. Kerrigan and M. Morari}}

@INPROCEEDINGS{Gondhalekar_Safe_Move-block_MPC,
  author={Gondhalekar, Ravi and Imura, Jun-ichi},
  booktitle={2007 46th IEEE Conference on Decision and Control}, 
  title={Recursive feasibility guarantees in move-blocking MPC}, 
  year={2007},
  volume={},
  number={},
  pages={1374-1379},
}

@article{Gondhalekar_Least_Restrictive_Move-block_MPC,
title = {Least-restrictive move-blocking model predictive control},
journal = {Automatica},
volume = {46},
number = {7},
pages = {1234-1240},
year = {2010},
issn = {0005-1098},
doi = {https://doi.org/10.1016/j.automatica.2010.04.010},
url = {https://www.sciencedirect.com/science/article/pii/S0005109810001676},
author = {Ravi Gondhalekar and Jun-ichi Imura},
}

@article{Son_semi-explicit_move-block_MPC,
title = {Move blocked model predictive control with improved optimality using semi-explicit approach for applying time-varying blocking structure},
journal = {Journal of Process Control},
volume = {92},
pages = {50-61},
year = {2020},
issn = {0959-1524},
doi = {https://doi.org/10.1016/j.jprocont.2020.04.002},
url = {https://www.sciencedirect.com/science/article/pii/S0959152420301992},
author = {Sang Hwan Son and Tae Hoon Oh and Jong Woo Kim and Jong Min Lee},
}

@article{Brunner_robust_self_trigger,
title = {Robust self-triggered MPC for constrained linear systems: A tube-based approach},
journal = {Automatica},
volume = {72},
pages = {73-83},
year = {2016},
issn = {0005-1098},
doi = {https://doi.org/10.1016/j.automatica.2016.05.004},
url = {https://www.sciencedirect.com/science/article/pii/S0005109816301881},
author = {Florian David Brunner and Maurice Heemels and Frank Allgöwer},
}

@INPROCEEDINGS{Zhan_robust_self_trigger,
  author={Zhan, Jingyuan and Li, Xiang and Jiang, Zhong-Ping},
  booktitle={2017 American Control Conference (ACC)}, 
  title={Self-triggered robust output feedback model predictive control of constrained linear systems}, 
  year={2017},
  volume={},
  number={},
  pages={3066-3071},
  doi={10.23919/ACC.2017.7963418}}

@ARTICLE{Dai_cost-based_self_trigger,
  author={Dai, Li and Cannon, Mark and Yang, Fuwen and Yan, Shuhao},
  journal={IEEE Transactions on Automatic Control}, 
  title={Fast Self-Triggered MPC for Constrained Linear Systems With Additive Disturbances}, 
  year={2021},
  volume={66},
  number={8},
  pages={3624-3637},
  doi={10.1109/TAC.2020.3022734}}

@ARTICLE{Lefevre_ACC_2016,
  author={Lefèvre, Stéphanie and Carvalho, Ashwin and Borrelli, Francesco},
  journal={IEEE Transactions on Automation Science and Engineering}, 
  title={A Learning-Based Framework for Velocity Control in Autonomous Driving}, 
  year={2016},
  volume={13},
  number={1},
  pages={32-42},
  }

@INPROCEEDINGS{Kim_ACC_2019,
  author={Kim, Yeojun and Guanetti, Jacopo and Borrelli, Francesco},
  booktitle={2019 18th European Control Conference (ECC)}, 
  title={Robust Eco Adaptive Cruise Control for Cooperative Vehicles}, 
  year={2019},
  volume={},
  number={},
  pages={1214-1219},
 }

@article{STEIN20112626,
title = {Implicit upper bound error estimates for combined expansive model and discretization adaptivity},
journal = {Computer Methods in Applied Mechanics and Engineering},
volume = {200},
number = {37},
pages = {2626-2638},
year = {2011},
issn = {0045-7825},
author = {Erwin Stein and Marcus Rüter and Stephan Ohnimus}
}

@INPROCEEDINGS{tomlin2017,
  author={Herbert, Sylvia L. and Chen, Mo and Han, SooJean and Bansal, Somil and Fisac, Jaime F. and Tomlin, Claire J.},
  booktitle={2017 IEEE 56th Annual Conference on Decision and Control (CDC)}, 
  title={FaSTrack: A modular framework for fast and guaranteed safe motion planning}, 
  year={2017},
  volume={},
  number={},
  pages={1517-1522}}

@article{VOELKER2013943,
title = {Moving horizon estimation: Error dynamics and bounding error sets for robust control},
journal = {Automatica},
volume = {49},
number = {4},
pages = {943-948},
year = {2013},
issn = {0005-1098},
author = {Anna Voelker and Konstantinos Kouramas and Efstratios N. Pistikopoulos}
}

@ARTICLE{hur2019,
  author={Kim, Jae-Kyeong and Ma, Jin and Sun, Kai and Lee, Jaegul and Shin, Jeonghoon and Kim, Yonghak and Hur, Kyeon},
  journal={IEEE Transactions on Power Systems}, 
  title={A Computationally Efficient Method for Bounding Impacts of Multiple Uncertain Parameters in Dynamic Load Models}, 
  year={2019},
  volume={34},
  number={2},
  pages={897-907},
}

@article{roy2010,
author = {Roy, Christopher},
year = {2010},
month = {01},
pages = {},
title = {Review of Discretization Error Estimators in Scientific Computing},
journal = {48th AIAA Aerospace Sciences Meeting Including the New Horizons Forum and Aerospace Exposition},
}

@article{Tarjan_SCC,
author = {Tarjan, Robert},
title = {Depth-First Search and Linear Graph Algorithms},
journal = {SIAM Journal on Computing},
volume = {1},
number = {2},
pages = {146-160},
year = {1972},
doi = {10.1137/0201010},
URL = {https://doi.org/10.1137/0201010},
eprint = {https://doi.org/10.1137/0201010}
}

\end{document}